\documentclass[onecolumn, a4size, 11pt,final]{IEEEtran}
\usepackage{amsmath}
\usepackage{amssymb}
\usepackage{amsfonts}
\usepackage[final]{graphicx}
\usepackage{ctable}

\renewcommand{\baselinestretch}{1.4}

\title{On Spatial Capacity of Wireless Ad Hoc Networks with  Threshold Based
Scheduling}
\author{Yue Ling Che, Rui Zhang, Yi Gong, and Lingjie Duan
\thanks{This work has been presented in part at \emph{IEEE International
Symposium on Information Theory (ISIT),  Istanbul, Turkey, July 7-12, 2013}.}
\thanks{Y.~L.~Che and Y.~Gong are with the School of Electrical
and Electronic Engineering, Nanyang Technological University,
Singapore (e-mail: chey0004@ntu.edu.sg; eygong@ntu.edu.sg).}
\thanks{R.~Zhang is with the Department of Electrical and Computer Engineering,
National
University of Singapore (e-mail: elezhang@nus.edu.sg). He is also with the
Institute for Infocomm
Research, A*STAR, Singapore.}
\thanks{L.~Duan is with the Engineering Systems and Design Pillar, Singapore
University of Technology and Design (e-mail: lingjie\_duan@sutd.edu.sg).}
}

\begin{document}
\maketitle \thispagestyle{empty}

\begin{abstract}

This paper  studies spatial capacity in a stochastic wireless {\it ad hoc}
network, where  multi-stage probing and data transmission are sequentially
performed. We propose a
novel signal-to-interference-ratio (SIR) threshold based scheduling scheme: by
starting with initial probing, each transmitter iteratively decides to further
probe or stay idle, depending on whether the estimated SIR in the proceeding
probing is larger or  smaller than a predefined threshold. Since only local SIR
information is required for making transmission decision, the proposed scheme is
appropriate for distributed implementation in practical wireless ad hoc
networks. Although one can assume that the transmitters are initially deployed
according to a homogeneous Poisson point process (PPP), the SIR based scheduling
makes the PPP no longer applicable to model the locations of retained
transmitters in the subsequent probing and data transmission phases, 
due to the interference induced coupling in their decisions. As the
analysis becomes very
complicated, we first focus on  single-stage probing  and find that
when the SIR threshold is set sufficiently small to assure an acceptable
interference level in the network, the proposed scheme can greatly outperform
the non-scheduling reference scheme in terms of spatial capacity.  We clearly
characterize the spatial capacity and obtain exact/approximate closed-form
expressions, by proposing a new approximate approach to deal with the correlated
SIR distributions over non-Poisson point processes.
Then we successfully extend to multi-stage probing by properly designing the
multiple SIR thresholds to assure gradual improvement of the spatial capacity.
Furthermore, we  analyze the impact of multi-stage probing overhead and present
a probing-capacity tradeoff in scheduling design. Finally, extensive numerical
results are  presented to demonstrate the performance of the proposed scheduling 
as compared to existing schemes.

\end{abstract}


\setlength{\baselineskip}{1.3\baselineskip}
\newtheorem{definition}{\underline{Definition}}[section]
\newtheorem{fact}{Fact}
\newtheorem{assumption}{Assumption}
\newtheorem{theorem}{\underline{Theorem}}[section]
\newtheorem{lemma}{\underline{Lemma}}[section]
\newtheorem{corollary}{\underline{Corollary}}[section]
\newtheorem{proposition}{\underline{Proposition}}[section]
\newtheorem{example}{\underline{Example}}[section]
\newtheorem{remark}{\underline{Remark}}[section]
\newcommand{\mv}[1]{\mbox{\boldmath{$ #1 $}}}
\newtheorem{property}{\underline{Property}}[section]

\section{Introduction}
Wireless \emph{ad hoc} networks have emerged as a promising technology  that can
provide seamless communication between wireless users (transmitter-receiver
pairs) without relying on any pre-existing infrastructure.
In such networks, the wireless users communicate with each other in a
distributed manner. Due to the lack of centralized coordinators to coordinate
the transmissions among the users, the wireless ad hoc network is  under
competitive and interference-dominant environment in nature. Thereby, efficient
transmission schemes for transmitters to effectively schedule/adapt their
transmissions are appealing for system performance improvement, and thus have
attracted wide research attentions in the past  decade.

Traditionally, each transmitter is enabled to independently decide whether to
transmit over a particular channel based on its own willingness  or channel
strength \cite{Andrews.COM.09}-\cite{FadingScheduling}, and the transmission
rate of each user can be  maximized by finding an optimal transmission
probability or an optimal  channel strength threshold, respectively. Although
easy to be implemented, such independent transmission schemes do not consider
the resulting user interactions in the wireless ad hoc networks due to the
co-channel interference, and thus do not achieve high system performance in
general cases.
Therefore, more complex transmission schemes have been proposed to exploit the
user interactions by exploring the  information of signal-to-interference-ratio
(SIR). For example, by iteratively adapting the transmit power level based on
the estimated SIR, the  Foschini-Miljanic algorithm \cite{Foschini-Miljanic}
assures zero outage probability and/or minimum aggregate power consumption for
uplink transmission in a cellular network. In \cite{Yates}, Yates has studied
power convergence conditions for such iterative power control algorithms.
Moreover, there have been some recent studies (e.g.  \cite{ElBatt.TWC.2004} and
\cite{Huang.TCOM.09}) that extend  the Foschini-Miljanic algorithm to the
wireless ad hoc network through joint scheduling and power control transmission
schemes.  In addition, by adapting the transmission probability depending on the
received SIR, \cite{Baccelli.Adaptive.13} has studied various
random access schemes to improve the system
throughput and/or the user fairness. However,  \cite{Foschini-Miljanic}-\cite{Baccelli.Adaptive.13}
either require  each transmitter to know
at least the wireless environment information of its neighbors, or are of high implementation
complexity, and thus are not appropriate for practical large-scale wireless ad hoc networks.

On the other hand,  due to the randomized location of each transmitter and the
effects of channel fading, the  network-level performance analysis is
fundamentally important for the study of wireless ad hoc networks.  It is noted
that Gupta and Kumar  in  \cite{Gupta.IT.2000} studied \emph{scaling laws},
which quantified the increase of the volume of capacity region  over the number
of transmitters in  ad hoc networks. Moreover, to determine   the set of active
transmitters that can yield maximum  aggregate Shannon capacity in the network,
the authors in   \cite{AndrewsDinitz.09}-\cite{Asgeirsson.11} addressed the
capacity maximization problem for an arbitrary wireless ad hoc network.
However,  \cite{Gupta.IT.2000}-\cite{Asgeirsson.11} did not consider the impact
of \emph{spatial configuration} of the ad hoc network, which is a critical factor that
determines the ad hoc network capacity~\cite{Andrews.Mag.10}. 
It came to our attention that as a powerful tool to   capture the impact
of wireless users' spatial randomness  on the network performance,
stochastic geometry  \cite{Stoyan.SG.95} is able to provide
more comprehensive characterization of the performance of wireless networks,
 and thus has attracted great attentions from both academy and industry \cite{Andrews.Mag.10}, \cite{Ghosh.Mag.12}.
Among all the tools provided by stochastic geometry, homogeneous Poisson point process (PPP) \cite{Kingman.book93}
is the most widely used one for network topology modeling and performance analysis. Under the assumption that
 the transmitters are deployed according to a homogeneous PPP,
the  exact/approximate capacity of a wireless ad hoc network under various
\emph{independent} transmission schemes, such as Aloha-based random transmission
\cite{Andrews.COM.09}, channel-inversion based power control
\cite{Jindal.TWC.08}, and channel-threshold based scheduling \cite{Weber.IT.07},
\cite{FadingScheduling},    can all be successfully characterized by using
advanced tools from stochastic geometry.
However, limited work based on stochastic geometry has studied  SIR-based transmission schemes, where the \emph{user interactions} are involved.
It is  noted that \cite{Kim.ProbScheld.14} studied a probability-based scheduling scheme, where each transmitter
independently adjusts its current transmission  probability based on the received SIR
in the proceeding iteration.
However,   \cite{Kim.ProbScheld.14}  only studied the convergence
of the probability-based  scheduling,  without addressing the network capacity with spatiality distributed users. 
To our best knowledge, there has been no existing work on studying the wireless ad hoc network capacity
with a SIR-based transmission scheme. Hence, the impact of SIR-based
transmissions is limitedly understood from the network-level point of view.

\emph{A principle goal of this study is to use stochastic geometry to fill the void of wireless  network
capacity characterization by an efficient  SIR-based transmission scheme.}
To this end, we propose a novel  SIR-threshold based  scheduling scheme for a
single-hop slotted wireless ad hoc network. We consider a probe-and-transmit
protocol,  where multi-stage probings are  sequentially performed to gradually
determine the  transmitters that are allowed to transmit data in each slot.
Specifically, we assume there are in total $N$ probing phases and one data transmission phase in each slot,
$1\leq N < \infty$. We sequentially label the $N$ probing phases  as  P-Phase $0$, P-Phase $1$,
..., and P-Phase $N-1$, and label the data transmission phase as D-Phase.
As   illustrated in Fig.~1, if the feedback SIR from  receiver $i$ in P-Phase 
$k-1$, $1\leq k \leq N-1$, is no smaller than a pre-defined threshold,
transmitter $i$ decides to transmit in P-Phase~$k$; otherwise, to improve the
system throughput as well as save its own energy, transmitter $i$  stays idle in
the remaining time of the slot as in \cite{TwoProbing}, so as to let other
transmitters that have higher SIR levels   re-contend the current  transmission
opportunity.
Since  each transmitter only requires  direct-channel SIR feedback from its
intended receiver for limited times, the proposed scheme is  appropriate for
distributed implementation in  practical wireless ad hoc networks.
In this paper,  we characterize the wireless ad hoc network capacity  with a
metric called {\it spatial capacity}, which has been used in \cite{Baccelli.IT.06}
and gives the average number of
successful transmitters per unit area for any given  initial transmitter
density. We aim  at  closed-form spatial  capacity characterization  and  maximization
by exploring the SIR-threshold based transmission.

\begin{figure}[t]
\centering
\DeclareGraphicsExtensions{.eps,.mps,.pdf,.jpg,.png}
\DeclareGraphicsRule{*}{eps}{*}{}
\includegraphics[angle=0, width=0.98\textwidth]{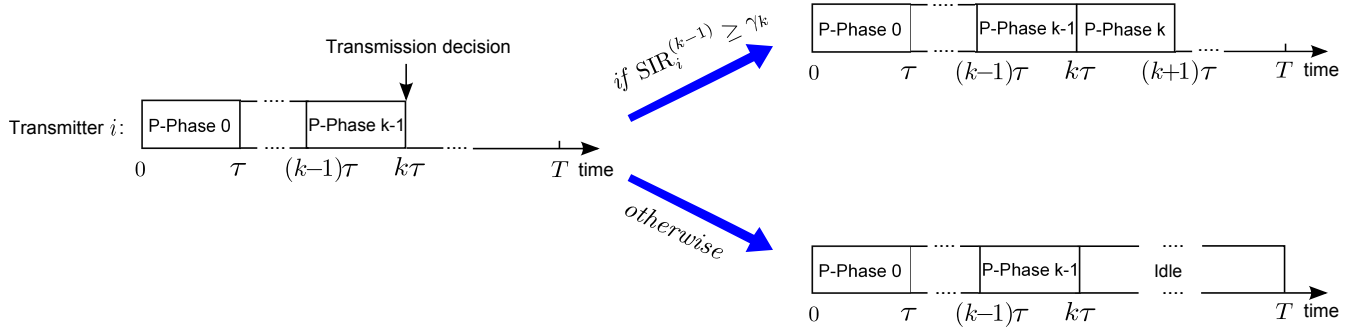}
\caption{Illustration of the SIR-threshold based transmission in P-Phase $k$: If
the  SIR in P-Phase~$k-1$  is no smaller than the threshold $\gamma_k\geq 0$ in
P-Phase~$k$, the transmitter decides to transmit in P-Phase~$k$; otherwise, it stays
idle in the remaining time of this slot. }

\label{fig: proposed_scheme}
\end{figure}

 The key contributions of this paper are summarized as follows.

\begin{itemize}
 \item \emph{Novel SIR-threshold based scheduling scheme:} In Section II, we
propose a  novel SIR-threshold based transmission scheme for a single-hop
wireless ad hoc network, which can be implemented efficiently in a distributed
manner. Though one can use  a   homogeneous PPP to model the stochastic
locations of the   transmitters in the initial probing phase, we find that due
to the iterative SIR-based scheduling, the PPP model is no longer applicable to
model the locations of the retained transmitters  in all the subsequent probing
or data transmission phases. Furthermore, since the  SIR distributions in all
the probing and data transmission phases are strongly correlated, it is
challenging to analyze/characterize the spatial capacity of the proposed scheme.

 \item \emph{Single-stage probing for  spacial capacity improvement:}
 In Section III, we start up with single-stage probing ($N=1$) to clearly decide
the SIR threshold for the proposed scheme and characterize the spatial capacity.
We show that a small SIR threshold can efficiently reduce the  retained
transmitter number and thus the interference level in the data transmission phase, while a large SIR threshold will overly reduce the retained transmitter
number and does not help improve the spatial capacity.  We also propose a new
approximate approach to characterize the spatial capacity in closed-form, which
is useful for analyzing performance of wireless networks with interacted
transmitters.

\item \emph{Multi-stage probing for spatial capacity improvement:}
In Section IV, we extend proposed scheduling scheme from the single-stage
probing ($N=1$) to multi-stage probing ($N>1$) for greater spatial capacity
improvement. We show that once a sequence of increasing SIR thresholds are
properly decided over probing phases, the spatial capacity is assured to
gradually improve. As multi-stage probing can introduce non-ignorable overhead
in each time slot, which  reduces the spatial capacity, we study an interesting
probing-capacity tradeoff over the probing-stage number $N$.

\item \emph{Performance evaluations for network design:} In both Section III and
Section IV, we also provide extensive numerical results  to further evaluate the
impact of  key  parameters of the proposed scheme. In particular, we present a
density-capacity tradeoff in Section III-C-1), which shows that a small initial
transmitter density can help improve the spatial capacity, while a large one
will introduce high interference level and thus reduce the spatial capacity.
To highlight the spatial capacity improvement performance of the 
proposed scheme, we also compare the proposed scheme with existing distributed 
scheduling schemes in Section III-C-2). Moreover, we consider a practical scenario with 
SIR estimation and feedback errors and show that the proposed design is robust to the SIR errors in Section III-C-3) by simulation. 
In Section IV, we study an example with $N=2$ and show the
corresponding spatial capacity over  both  SIR thresholds  in P-Phase~1 and 
D-Phase. Interestingly, our numerical results show that the
former SIR threshold plays  a more critical role in determining the spatial
capacity than  the latter one, since the  former SIR threshold decides how many
transmitters can  have a second chance to contend the transmission opportunity.

\end{itemize}

It is noted that some of the existing work has addressed the throughput/capacity analysis of  a wireless communication system from the information-theoretic point of view.
For example,  Tse and Hanly considered a multipoint-to-point system and characterized the
throughput capacity region and delay-limited capacity region of the fading multiple-access channel in \cite{Tse.IT.98} and \cite{Hanly.IT.98}, respectively,
where the optimal power and/or rate allocation that can achieve the boundary of the capacity regions was derived.
Although appealing, both \cite{Tse.IT.98} and \cite{Hanly.IT.98} have assumed multiuser detection at a centralized receiver and  ignored the impact of the random network topology driven by mobile transmitters and receivers’ mobility,
and thus cannot completely provide network-level system performance characterization with distributed single-user detection (i.e., treating the multiuser interference as noise)  receivers.
Unlike Tse and Hanly's works in \cite{Tse.IT.98} and \cite{Hanly.IT.98}, we use stochastic geometry to model the large-scale random wireless ad hoc network topology,
and novelly analyze the network-level performance of the iterative SIR-threshold based scheduling.

In addition, it is also noted that  some existing work has adopted tools from stochastic geometry to study the  non-PPP based wireless network.
For example, by  using a PPP to approximate the underlying non-PPP based spatial distribution of the transmitters' locations, \cite{Hasan.TWC.2007.GuardZone}-\cite{Song.JSAC} have successfully characterized the non-PPP based wireless network capacity.  Unlike \cite{Hasan.TWC.2007.GuardZone}-\cite{Song.JSAC}, due to the iterative SIR-based scheduling of the proposed scheme, we need to address not only the non-PPP based spatial distribution of the transmitters' locations,  but also  the resulting strongly-correlated SIR distributions
over all probing and data transmission phases. To our best knowledge, such correlated SIR analysis/chracterization in non-PPP based wireless networks has not been addressed in the existing work based on stochastic geometry.

\section{System Model and Performance Metric} \label{section: SC_model}
In this section, we  describe the considered transmission schemes in this paper.
We then  develop the network model based on stochastic geometry. At last, we
define the  spatial capacity as our performance metric.

\subsection{Transmission Schemes}
We focus on the proposed scheme with SIR-threshold based scheduling. For
comparison, we also consider a reference scheme without any transmission
scheduling.
For both transmission schemes, we assume that all transmitters  transmit in a
synchronized time-slotted manner. We also  assume  that all  transmitters
transmit at the same power level,\footnote{In general, each transmitters can
transmit at different power levels by iteratively adjusting its transmit power
based on the feedback SIR information, as in \cite{Foschini-Miljanic} or
\cite{Yates}. However, in this paper,  we mainly focus on SIR-based transmission
scheduling and thus restrict transmit power adaptation  to be binary for
simplicity.} which is normalized to be unity for convenience.

\subsubsection{SIR-Threshold Based Scheme}
Based on the probe-and-transmit protocol, in each time slot, $N$ probing phases  
with $1\leq N < \infty$ are sequentially implemented before the data transmission phase. 
We assume $N$ is a pre-given parameter and its effects will
be studied later in Section IV-B.
Moreover, as shown in Fig.~1, we denote the duration of a  time slot and a
probing phase  as $T$ and $\tau$, respectively, with $\tau\ll T$, such that
$N\tau<T$, as in \cite{TwoProbing}.
By normalizing over  $T$,   the \emph{effective data transmission time} in a
time slot is obtained as $\frac{T-N\tau}{T}$, which reduces linearly over $N$
\cite{Hossian.book}.  Furthermore, we assume if a transmitter transmits probing
signals in a probing phase, its intended receiver is able to measure the
received signal power over the total interference power, i.e., the  SIR, and
feeds it back to the transmitter at the end of the probing phase.
The specific algorithm design on SIR estimation and feedback is out of the scope
of this paper and is not our focus.
To obtain tractable analysis,  we assume  perfect SIR estimation and  feedback in this paper, and thus the SIR value is exactly
known at the transmitter; however, the impact of finite SIR estimation and feedback errors on the network capacity is important to practical design and thus will also be evaluated by simulation.  

According to the feedback SIR level of its own channel,
each transmitter  iteratively performs the threshold-based transmission
decision in each P-Phase or D-Phase, for which the details are
given as follows:
\begin{itemize}
\item In the initial probing phase, i.e.,  P-Phase $0$, to initialize the
communication between each transmitter and  receiver pair, all transmitters
independently transmit   probing signals to their intended receivers. 
Each receiver then estimates the channel amplitude and phase (for possible
coherent communication in the subsequent probing  and data transmission phases),
and measures the received SIR of the probing signal. Each transmitter receives
the feedback SIR from its intended receiver at the end of P-Phase $0$.

\item  In each of the remaining probing phases from P-Phase $1$ to P-Phase $N-1$, by
exploiting the feedback SIR in the proceeding probing phase, each transmitter
decides whether to transmit in the current probing phase with a predefined SIR
threshold. Specifically, suppose a transmitter  transmits in P-Phase $k-1$, $1\leq
k \leq N-1$. As shown in Fig.~1, if the feedback SIR  in P-Phase $k-1$  is larger
than or equal to the predefined SIR-threshold, denoted by $\gamma_{k}\geq 0$ for
P-Phase $k$, the transmitter   continues its transmission in P-Phase $k$ and thus
receives the feedback SIR in P-Phase $k$; otherwise, to improve the system
throughput as well as save its energy, the transmitter  decides not to transmit
any more in the remaining time of this  slot and will seek another transmission
opportunity in the next  slot, so as to let other transmitters that have higher
SIR levels to re-contend the current transmission opportunity.

\item In the D-Phase, similar to the SIR-threshold based
scheduling from P-Phase $1$ to P-Phase $N-1$, if a transmitter transmits in P-Phase $N-1$ and its  feedback SIR in P-Phase $N-1$ is larger than or equal to the
predefined threshold, denoted by $\gamma_N\geq 0$ for the D-Phase,  the transmitter sends data to its intended receiver; otherwise, the
transmitter remains silent in the rest time of this slot. The data transmission
is successful if the   SIR  at the receiver is larger than or equal to the
required SIR level, denoted by $\beta>0$.
\end{itemize}

\subsubsection{Reference Scheme}
There is no transmission scheduling in the reference scheme. In each time slot,
we assume all transmitters transmit data directly to their intended receivers in an
independent manner. Thus, the effective data transmission time for the reference
scheme is $1$.\footnote{It is worth pointing out that for the reference scheme,
an initial training is needed prior to data transmission for the receiver to
estimate the channel for coherent communication, similar to the initial probing
of the proposed scheme with $N=1$, but without the SIR feedback to the
transmitter. Here, we have assumed that such training incurs a negligible time
overhead as compared to each slot duration.}     The data transmission is
successful if  the SIR at the receiver  is larger than or equal to the  required
SIR level $\beta$ as the proposed scheme.
Note that by implementing an initial probing phase before the data transmission,
the reference scheme can be improved to be a proposed scheme with single-stage
probing.

\subsection{Network Model}
In the next, we develop the network model based on stochastic geometry. For both
considered transmission schemes, we focus on single-hop communication in one
particular time slot.

For both schemes, we assume that all transmitters are independently and uniformly distributed
in the unbounded two-dimensional plane $\mathbb{R}^2$.
We thus model the
locations of all  the transmitters by a  homogeneous PPP with density $\lambda$. 
Due to the lack of central infrastructure for coordination in the wireless ad hoc network, we assume the transmitters have no knowledge about their surrounding wireless environment, and thus  intend to transmit independently in a time slot
with probability $\theta\in (0,1)$, as in \cite{Andrews.COM.09}-\cite{FadingScheduling}.
Denote $\lambda_0=\lambda\theta$ as the density of the initial transmitters that have the intention
to transmit in a particular time slot.
According to the Coloring theory \cite{Stoyan.SG.95}, the process of the initial transmitters
for both schemes is a homogeneous PPP with density  $\lambda_0$, which is denoted by $\Phi_0$.
Without loss of generality, we assume $\lambda$ and $\theta$ and hence $\lambda_0$ are given parameters, and will discuss the effects of $\lambda_0$ later in Section III-C. 
We assume each transmitter has one intended receiver,
which  is  uniformly distributed on a circle of radius $d$ meters
(m) centered at the transmitter.
We denote the locations of the $i$-th transmitter and its intended receiver as
$x_i$, with $x_i \in \Phi_0$, and $r_i$ (not included in $\Phi_0$),
respectively. The path loss between the $i$-th transmitter  and the $j$-th
receiver is given by $l_{ij}=|x_i-r_j|^{-\alpha}$, where $\alpha\!>\!2$ is the
path-loss exponent.
We use $h_{ij}$ to denote the distance-independent channel fading coefficient
from transmitter $i$ to receiver $j$. We assume flat Rayleigh fading, where all
$h_{ij}$'s are  independent and  exponentially distributed random variables with
unit mean. We also assume that $h_{ij}$'s do not change within one time-slot. We
denote the SIR at the $i$-th receiver  as $\textrm{SIR}_i^{(0)}$, which   is
given by
\begin{equation}
\textrm{SIR}_i^{(0)}=\frac{h_{ii}d^{-\alpha}}{\sum_{x_j\in \Phi_0, j\neq i}
h_{ji}l_{ji}}. \label{eq: SIR_0_pilot}
\end{equation}
Note that for the reference scheme without  transmission scheduling,
$\textrm{SIR}_i^{(0)}$ gives the received SIR level at the  $i$-th receiver for
the data transmission of transmitter $i$. As a result, in the reference scheme,
the data transmission of transmitter $i$  is \emph{successful}  if
$\textrm{SIR}_i^{(0)}\!\geq \!\beta$ is satisfied.

Unlike the reference scheme, in the proposed scheme, $\textrm{SIR}_i^{(0)}$
only gives the received SIR level at the  $i$-th receiver in the initial probing
phase P-Phase~$0$. We then denote the point process formed by the retained
transmitters in P-Phase $k$ with $1\leq k \leq N-1$, or the D-Phase with $k=N$, as $\Phi_k$. 
We also denote  $\textrm{SIR}_i^{(k)}$ as the
received SIR at the $i$-th receiver in  $\Phi_k$.  Clearly, we have
$\Phi_k=\{x_i\in \Phi_{k-1}: \textrm{SIR}_i^{(k-1)}\geq \gamma_k\}$, where the
number of transmitters in $\Phi_k$ is reduced as compared to that in
$\Phi_{k-1}$. Thus, it is easy to verify that  $\textrm{SIR}_i^{(k)}\!\geq
\!\textrm{SIR}_i^{(k-1)}$ for  any given $\gamma_k\!\geq \!0$, $\forall i\in
\Phi_{k-1} \cap \Phi_{k}$.  Moreover, similar to $\textrm{SIR}_i^{(0)}$, for any
$\Phi_k$, $k\in \{1,..., N\}$,   we can express  $\textrm{SIR}_i^{(k)}$ as
\begin{equation}
\textrm{SIR}_i^{(k)}=\frac{h_{ii}d^{-\alpha}}{\sum_{x_j\in \Phi_k,  j\neq i}
h_{ji}l_{ji}},~~ k\in \{1,...,N\}. \label{eq: SIR_k}
\end{equation}
It is worth noting that  due to the SIR-based scheduling, the transmitters are
not retained independently in $\Phi_k$. Thus, unlike $\textrm{SIR}_i^{(0)}$  in
(\ref{eq: SIR_0_pilot}), which is  determined by the homogeneous PPP $\Phi_0$,
$\textrm{SIR}_i^{(k)}$ in  (\ref{eq: SIR_k})  is determined by the
\emph{non-PPP} $\Phi_k$  in general \cite{Stoyan.SG.95}.  For the proposed
scheme,  the data  transmission of transmitter~$i$  is \emph{successful} if
$\textrm{SIR}_i^{(k-1)}\!\geq \!\gamma_k$, $\forall k\in \{1,..., N\}$, and
$\textrm{SIR}_i^{(N)}\! \geq \!\beta$ are all satisfied.

\subsection{Spatial Capacity}
Due to the stationarity of the homogeneous PPP $\Phi_0$,  it is easy to verify
that $\Phi_k$, $\forall k\in \{1,...N\}$,  is also stationary
\cite{FadingScheduling}.
We thus consider a typical  pair of transmitter and receiver in this paper.
Without loss of generality, we  assume that the typical  receiver is  located
at the origin. The typical pair of transmitter and receiver is named pair 0,
i.e., $i=0$.
Denote the {\it successful transmission probability} of the typical pair in the
data transmission phase of the proposed scheme with $N$ probing phases or  the
reference  scheme as $\mathcal{P}_{0}^{p,N}$ or $\mathcal{P}_{0}^{r}$,
respectively. We thus have
\begin{align}
\mathcal{P}_{0}^{p,N}&= \mathbb{P}(\textrm{SIR}_0^{(0)}\geq
\gamma_1,...,\textrm{SIR}_0^{(N-1)}\geq \gamma_N, \textrm{SIR}_0^{(N)}\geq
\beta).  \label{eq: P_succ_p_def} \\
\mathcal{P}_{0}^{r}&= \mathbb{P}(\textrm{SIR}_0^{(0)}\geq \beta).  \label{eq:
P_succ_r_def}
\end{align}

We adopt  {\it spatial capacity} as our performance metric, which is defined as
the spatial  density of successful transmissions, or more specifically the
average number of transmitters with successful data transmission  per unit area.
Considering the effective data transmission time in a time slot,
we thus define the spatial capacity by the proposed scheme with $N$ probing phases  
and the reference scheme as $\mathcal{C}^{p,N}$ and $\mathcal{C}^{r}$,
respectively, given by
\begin{align}
\mathcal{C}^{p,N} \triangleq & \frac{T-N\tau}{T} \lambda_0
\mathcal{P}_{0}^{p,N},  \label{eq: C_p_def} \\
\mathcal{C}^{r} \triangleq & \lambda_0 \mathcal{P}_{0}^{r}.  \label{eq: C_r_def}
\end{align}

For the reference scheme, it is noted that $\mathcal{P}^{r}$, given in
(\ref{eq: P_succ_r_def}), is the complementary cumulative distribution function
(CCDF) of $\textrm{SIR}_0^{(0)}$ taken at the value of  $\beta$. We then have
the following proposition.
\begin{proposition} \label{proposition: SIR_1}
The successful transmission probability   in the reference  scheme is
\begin{equation}
\mathcal{P}_{0}^{r}= \exp (-\pi \lambda_0 d^2 \beta^{\frac{2}{\alpha}} \rho ),
\label{eq: suc_prob_stage_1}
\end{equation}
where $\rho =\int_{0}^{\infty}\! \frac{1}{1+v^{\alpha\!/\!2}}\, dv$.
When  $\alpha=4$, we have $\rho = \frac{\pi}{2}$.
\end{proposition}

The proof of Proposition \ref{proposition: SIR_1} is similar to that of [30,
Theorem~2], which is based on the probability generating functional  (PGFL) of
the PPP,   and thus is omitted here.

Since the network interference level in the D-Phase increases
over the initial transmitter density $\lambda_0$,  we find that
$\mathcal{P}_{0}^{r}$ in (\ref{eq: suc_prob_stage_1})  monotonically decreases
over $\lambda_0$ as expected. Moreover, from (\ref{eq: C_r_def}) and  (\ref{eq:
suc_prob_stage_1}), we can  obtain the expression of $\mathcal{C}^{r}$ as
\begin{equation}
\mathcal{C}^{r}=  \lambda_0  \exp (-\pi \lambda_0 d^2 \beta^{\frac{2}{\alpha}}
\rho ). \label{eq: C_r}
\end{equation}
It is observed from  (\ref{eq: C_r})  that  unlike $\mathcal{P}_{0}^{r}$,   the
spatial capacity $\mathcal{C}^{r}$    does not vary monotonically over
$\lambda_0$, since $\mathcal{C}^{r}$ can be benefited by increasing $\lambda_0$
if the resulting interference is acceptable. Moreover, from (\ref{eq:
suc_prob_stage_1}) and (\ref{eq: C_r}), it is also expected  that  both
$\mathcal{P}_{0}^{r}$ and  $\mathcal{C}^{r}$ monotonically decrease  over  the
distance $d$ between each transmitter and receiver pair, due to the reduced
signal power received at the receiver, and decrease over the required SIR level
$\beta$.

Unlike the reference scheme, which is determined by the homogeneous PPP
$\Phi_0$, the proposed scheme is jointly determined by $\Phi_0$ and a sequence
of non-PPPs $\{\Phi_k\}$, $1\leq k\leq N$, where the resulting SIR distributions
are correlated. Therefore, it is very difficult to analyze/characterize the
spatial capacity of the proposed scheme with $N$ probing phases. To start up,
in the next section, we focus on a simple  case with single-stage probing
($N=1$) for some insightful results.

\section{SIR-Threshold based Scheme with Single-Stage Probing}
In this section, we consider the proposed scheme with single-stage probing,
i.e., $N=1$. In this case, there is only one round of SIR-based scheduling,
which is implemented with the threshold $\gamma_1$.
For notational simplicity, for the case of $N=1$, we omit the superscript $N$
and use $\mathcal{P}_{0}^{p}$  and $\mathcal{C}^{p}$ to represent the successful
transmission probability and the spatial capacity of the typical transmitter,
respectively.
Based on (\ref{eq: P_succ_p_def}),  the successful transmission probability for
the case of $N=1$ is reduced to
\begin{align}
\mathcal{P}_{0}^{p}&= \mathbb{P}(\textrm{SIR}_0^{(0)}\geq \gamma_1,
\textrm{SIR}_0^{(1)}\geq \beta)  \label{eq: P_succ_p_1} \\
 &=\mathbb{P}(\textrm{SIR}_0^{(0)}\geq \gamma_1) \mathbb{P}(
\textrm{SIR}_0^{(1)}\geq \beta| \textrm{SIR}_0^{(0)}\geq \gamma_1). \label{eq:
P_succ_p_1_cond}
\end{align}
Moreover, when $N=1$,  the effective data transmission time for the proposed
scheme is $\frac{T-\tau}{T}$. Since $\tau\ll T$,  we assume  the single-stage
probing overhead is negligible; and thus, the effective data transmission time
becomes 1 as the reference scheme. Consequently, based on (\ref{eq: C_p_def}),
we can express the spatial capacity  $\mathcal{C}_{0}^{p}$ as
\begin{align}
\mathcal{C}^{p} = & \lambda_0 \mathcal{P}_{0}^{p}.  \label{eq: C_p_1}
\end{align}
Furthermore, by substituting  (\ref{eq: P_succ_p_1_cond}) to (\ref{eq: C_p_1}),
we  can express $\mathcal{C}^{p}$ alternatively as
\begin{align}
\mathcal{C}^{p} =&  \lambda_0 \mathbb{P}(\textrm{SIR}_0^{(0)}\geq \gamma_1)
\mathbb{P}( \textrm{SIR}_0^{(1)}\geq \beta| \textrm{SIR}_0^{(0)}\geq \gamma_1)
\nonumber \\
=&  \lambda_1 \mathbb{P}( \textrm{SIR}_0^{(1)}\geq \beta|
\textrm{SIR}_0^{(0)}\geq \gamma_1)
\end{align}
where $\lambda_1\!=\!\lambda_0 \mathbb{P}(\textrm{SIR}_0^{(0)}\!\geq
\!\gamma_1)$ is the density of $\Phi_1$ in the D-Phase, with
$\lambda_1\!\leq \!\lambda_0$. Based on Proposition \ref{proposition: SIR_1}, by
replacing $\beta$ with $\gamma_1$, it is easy to find that
\begin{equation}
\lambda_1 =  \lambda_0 \exp\big(-\pi \lambda_0 d^2
\gamma_1^{\frac{2}{\alpha}}\rho\big). \label{eq: lambda_1_expression}
\end{equation}

In the following two subsections, we  compare the spatial capacity of the two
considered schemes, and  characterize $\mathcal{C}^{p}$ for the proposed scheme.

\subsection{Spatial Capacity Comparison and Closed-form Characterization with
$\gamma_1=0$ and $\gamma_1\geq \beta$} \label{section: exact_analysis_single}
In this subsection, we  compare the spatial capacity of the proposed scheme with
that of the reference scheme.  We then characterize the spatial capacity
$\mathcal{C}^{p}$ for the proposed scheme and obtain closed-form expressions for
the cases of $\gamma_1=0$ and $\gamma_1\geq \beta$.

First, from (\ref{eq: C_r_def}) and (\ref{eq: C_p_1}), to compare
$\mathcal{C}^{p}$ and $\mathcal{C}^{r}$, the key is to compare
$\mathcal{P}_{0}^{p}$  and $\mathcal{P}_{0}^{r}$. In the reference  scheme,
denote the total interference power received at the typical receiver  as
$I_0=\sum_{x_i\in \Phi_0,  i\neq 0} h_{i0}l_{i0}$. In the proposed scheme,  the
received total interference power at the typical receiver in P-Phase~0 is thus
$I_0$, while that in the D-Phase is given by $I_1=\sum_{x_i\in
\Phi_1,  i\neq 0} h_{i0}l_{i0}$.
For any $\gamma_1\geq 0$,  we have $I_0\geq I_1$ since  $\Phi_1 \subseteq
\Phi_0$,  and thus $\textrm{SIR}_i^{(1)}\geq \textrm{SIR}_i^{(0)}$. As a result,
by changing over
the value of $\gamma_1 \! \in  \![0,\infty)$, we obtain the following
proposition.

\begin{proposition} \label{proposition: relationship_C1_C2}

Given the  required SIR level $\beta \! >  \! 0$, for any $\gamma_1  \! \in
\![0,\infty)$, we have
\begin{equation}
  \left\{
   \begin{array}{l}
    \mathcal{C}^{p} > \mathcal{C}^{r},~\text{if}~0<\gamma_1<\beta
\textrm{~\big(\emph{conservative} transmission regime\big)}\\

\mathcal{C}^{p}=\mathcal{C}^{r},~\text{if}~\gamma_1=0~\text{or}~\gamma_1=\beta
\textrm{~\big(\emph{neutral} transmission regime\big)}\\
    \mathcal{C}^{p} < \mathcal{C}^{r},~\text{if}~\gamma_1>\beta
\textrm{~\big(\emph{aggressive} transmission regime\big)}.
   \end{array}
  \right.   \label{spatial throughput_tau}
\end{equation}
\end{proposition}
\begin{proof}
Please refer to Appendix A.
\end{proof}

\begin{remark}
 Compared to the spatial capacity of the reference  scheme, Proposition
\ref{proposition: relationship_C1_C2} shows that for the proposed scheme with
SIR-threshold based scheduling,  due to the reduced interference level in the
D-Phase, the spatial capacity is improved in the conservative
transmission regime with  $0\!<\! \gamma_1\!<\!\beta$. However,  in the case of
the aggressive transmission regime with $\gamma_1\!>\!\beta$ , where the
transmitters that are able to transmit successfully in the D-Phase may also be removed from transmission, the  retained transmitters in the
D-Phase are  {\it overly} reduced. Consequently, the spatial
capacity is reduced in the aggressive transmission regime.  It is also noted
that in the  neutral transmission regime with  $\gamma_1\!=\!0$ or
$\gamma_1\!=\!\beta$,  the spatial capacity is identical for the two  schemes.
At last,  it is   worth noting that Proposition \ref{proposition:
relationship_C1_C2} holds regardless of
the specific channel fading distribution and/or transmitter location
distribution.
\end{remark}

Next, we  characterize the spatial capacity $ \mathcal{C}^{p}$ for the proposed
scheme with $N=1$. We focus on deriving the successful transmission probability
$\mathcal{P}_{0}^{p}$  in  (\ref{eq: P_succ_p_1}). Unlike  $\mathcal{P}_{0}^{r}$
 in (\ref{eq: P_succ_r_def}), which is given by the marginal CCDF of
$\textrm{SIR}_0^{(0)}$ taken at value $\beta$, $\mathcal{P}_{0}^{p}$ is given by
the \emph{joint} CCDF of  $\textrm{SIR}_0^{(0)}$ and $\textrm{SIR}_0^{(1)}$
taken at values $(\gamma_1,\beta)$.
In the following, we consider  three cases   $\gamma_1 =0 $,  $\gamma_1
\geq\beta$, and $0<\gamma_1<\beta$, and find closed-form spatial capacity
expressions for both cases of  $\gamma_1 =0 $ and  $\gamma_1 \geq\beta$.

Specifically, for the simple  case with $\gamma_1=0$, we can infer from
Proposition \ref{proposition: relationship_C1_C2} directly that
$\mathcal{C}^{p}=\mathcal{C}^{r}$, which is given in (\ref{eq: C_r}). For the
case of  $\gamma_1 \geq\beta$, since
$\textrm{SIR}_0^{(1)}>\textrm{SIR}_0^{(0)}$, we have $\mathbb{P}(
\textrm{SIR}_0^{(1)}\geq \beta| \textrm{SIR}_0^{(0)}\geq \gamma_1)=1$.
According to (\ref{eq: P_succ_p_1_cond}), we thus obtain
$\mathcal{P}_{0}^{p}=\mathbb{P}(\textrm{SIR}_0^{(0)}\geq \gamma_1)$ in this
case. By replacing  $\beta$ with $\gamma_1$ in Proposition \ref{proposition:
SIR_1}, we further obtain that $\mathcal{P}_{0}^{p}=\exp (-\pi \lambda_0 d^2
\gamma_1^{\frac{2}{\alpha}} \rho)$.  As a result,  based on (\ref{eq: C_p_1}),
we can express $ \mathcal{C}^{p}$ for the case of $\gamma_1 \geq\beta$ as
\begin{equation}
 \mathcal{C}^{p}=\lambda_0  \exp (-\pi \lambda_0 d^2 \gamma_1^{\frac{2}{\alpha}}
\rho ). \label{eq: SC_gamma_higher_beta}
\end{equation}
Similar to $ \mathcal{C}^{r}$ given in (\ref{eq: C_r}), it is observed that $
\mathcal{C}^{p}$ for both cases of $\gamma_1 =0 $ and  $\gamma_1 \geq\beta$ does
not vary monotonically over $\lambda_0$, but monotonically decreases over the
distance $d$ between each transmitter and receiver pair. Moreover, unlike  $
\mathcal{C}^{r}$  and $ \mathcal{C}^{p}$ for $\gamma_1=0$, $ \mathcal{C}^{p}$
for $\gamma_1\geq \beta$  is not related to the required SIR level $\beta$ any
more, since all the retained transmitters in the D-Phase meet
the condition  $\textrm{SIR}_i^{(1)}\geq \beta$ in this case.

However, for the case of $0\!<\! \gamma_1 \!<\! \beta$, $\mathcal{P}_{0}^{p}$
cannot be simply expressed by a marginal CCDF of $\textrm{SIR}_0^{(0)}$ as in
the above two cases. Moreover, from (\ref{eq: P_succ_p_1}), due to the
correlation between $\textrm{SIR}_0^{(0)}$ and $\textrm{SIR}_0^{(1)}$ as well as
the underlying non-PPP $\Phi_1$ that determines $\textrm{SIR}_0^{(1)}$, it is
very difficult, if not impossible, to find an exact expression of
$\mathcal{P}_{0}^{p}$ and thus $ \mathcal{C}^{p}$ in this case.  As a result, in
the next subsection, we  focus on finding a tight approximate to  $
\mathcal{C}^{p}$  with a tractable expression for the case of $0\!<\! \gamma_1
\!<\! \beta$.

\subsection{Approximate Approaches for Spatial Capacity Characterization with
$0<\gamma_1<\beta$} \label{section: spatial_capacity_characterization}
This subsection focuses on approximating the  spatial capacity of the proposed
scheme for the case of  $0\!<\! \gamma_1 \!<\! \beta$. We first propose a new
approximate approach for $ \mathcal{C}^{p}$ and obtain  an integral-based
expression.  Next, to find a closed-form expression for $ \mathcal{C}^{p}$, we
further approximate the integral-based expression obtained by the proposed
approach. At last, we apply the conventional approximate approach in the
literature and discuss its approximate performance.
The details of the three approximate approaches  are given as follows.

\subsubsection{Proposed   Approximation}
From (\ref{eq: P_succ_p_1}), to find a good approximate to $\mathcal{P}_{0}^{p}$
and thus $\mathcal{C}^{p}$, the key is to find a good approximate to the joint
SIR distributions in $\Phi_0$ and $\Phi_1$. Since $\Phi_1 \subseteq \Phi_0$,  we
first   divide the initial PPP $\Phi_0$ into two \emph{disjoint non-PPPs}: one
is $\Phi_1$, and the other is its complementary set $\Phi_1^{c}=\Phi_0-\Phi_1$,
which is the point process formed by the non-retained transmitters in the D-Phase. We denote the density of $\Phi_1^{c}$ as
$\lambda_1^c\!=\!\lambda_0-\lambda_1$. Clearly, $\Phi_1$ and $\Phi_1^c$ are
mutually dependent. Denote  the received SIR level at the typical receiver in
$\Phi_1^{c}$ as $\textrm{SIR}_0^{(1,c)}=h_{00}d^{-\alpha}/\sum_{i\in \Phi_1^c}
h_{i0}l_{i0}$. Since $\Phi_1\cup \Phi_1^{c}\!=\!\Phi_0$ and $\Phi_1\cap
\Phi_1^{c}\!=\!\emptyset$, we have
$1/\textrm{SIR}_0^{(0)}=1\big/\big(\textrm{SIR}_0^{(1)}+\textrm{SIR}_0^{(1,c)}
\big)$. As a result,
(\ref{eq: P_succ_p_1}) can be equally represented   by using the joint
distributions of $\textrm{SIR}_0^{(1)}$ and  $\textrm{SIR}_0^{(1,c)}$.

Next, we  state an assumption,  based on which  we can use a homogeneous PPP to
approximate $\Phi_1$ and $\Phi_1^{c}$, respectively, such that the existing
results on PPP interference distribution in the literature can be applied to
approximate the joint distributions of $\textrm{SIR}_0^{(1)}$ and
$\textrm{SIR}_0^{(1,c)}$.
\begin{assumption}
In the proposed scheme with $N=1$, the transmitters are retained independently
in the D-Phase, with probability
$\mathbb{P}(\textrm{SIR}_0^{(0)}\!\geq \! \gamma_1)$.
\end{assumption}

By applying Assumption 1, we denote  the resulting point processes formed by the
retained and non-retained transmitters in the D-Phase as
$\hat{\Phi}_1$ and $\hat{\Phi}_1^{c}$, respectively. Clearly, both
$\hat{\Phi}_1$ and $\hat{\Phi}_1^{c}$ are homogeneous PPPs. Moreover, the
density of $\hat{\Phi}_1$ or $\hat{\Phi}_1^{c}$ is the same as that of $\Phi_1$
or $\Phi_1^c$, respectively. Since the two homogeneous PPPs $\hat{\Phi}_1$ and
$\hat{\Phi}_1^c$ are disjoint, they are \emph{independent} of each other
\cite{Stoyan.SG.95}. Denote
$\hat{I}_1\!=\!\sum_{i\in\hat{\Phi}_1}\!h_{i0}l_{i0}$ and
$\hat{I}_1^c\!=\!\sum_{i\in\hat{\Phi}_1^{c}}\!h_{i0}l_{i0}$ as the received
interference power at the typical receiver in $\hat{\Phi}_1$ and
$\hat{\Phi}_1^{c}$, respectively. We then use  $f_{\hat{I}_1}(x_1)$ and
$f_{\hat{I}_1^c}(x_2)$ to denote the  probability density functions (pdfs) of
$\hat{I}_1$ and $\hat{I}_1^c$, respectively. The following lemma gives the
general interference pdf in a homogeneous
PPP-based network with Rayleigh fading channels, which is a  well-known result
in the literature (e.g., \cite{Haenggi.book}).

\begin{lemma}\label{lemma: interference_pdf}
 For any homogeneous PPP of density $\lambda\geq 0$, if the channel fading is
Rayleigh distributed, the pdf of the received interference $I$ at the typical
receiver is given by
\begin{equation}
f_I(x)\!=\!\frac{1}{\pi x} \!\sum_{i=1}^{\infty}
\!\frac{(-1)^{i\!+\!1}\Gamma(1\!+\! 2i/ \alpha)\sin(  2\pi i / \alpha)}{i!}
\bigg(\!\frac{\lambda \pi^2 2/ \alpha}{x^{2/ \alpha}\sin( 2\pi / \alpha)}\!
\bigg)^{\!i}. \label{eq: I_pdf_alphaNot4}
\end{equation}
Moreover, when $\alpha=4$, (\ref{eq: I_pdf_alphaNot4}) can be further expressed
in a simpler closed-form  as
\begin{equation}
f_I(x)=\frac{\lambda}{4} \Big(\frac{\pi}{x}\Big)^{3/2}\exp\Big(-\frac{\pi^4
\lambda^2}{16 x}\Big). \label{eq: I_pdf_alpha4}
\end{equation}
\end{lemma}

As a result, based on Lemma \ref{lemma: interference_pdf}, by  substituting
$\lambda\!=\!\lambda_1$  to (\ref{eq: I_pdf_alphaNot4}) and (\ref{eq:
I_pdf_alpha4}), we can obtain $f_{\hat{I}_1}(x_1)$ for  the cases of general
$\alpha$ and $\alpha\!=\!4$, respectively. Similarly, with
$\lambda\!=\!\lambda_1^c$, from (\ref{eq: I_pdf_alphaNot4}) and (\ref{eq:
I_pdf_alpha4}) we can obtain $f_{\hat{I}_1^c}(x_2)$ for  general $\alpha$ and
$\alpha\!=\!4$, respectively. Therefore, by approximating  $\Phi_1$ and
$\Phi_1^{c}$ by $\hat{\Phi}_1$ and $\hat{\Phi}_1^{c}$, respectively, we can
easily approximate the joint  distribution of $\textrm{SIR}_0^{(1)}$ and
$\textrm{SIR}_0^{(1,c)}$ based on the interference pdfs $f_{\hat{I}_1}(x_1)$ and
$f_{\hat{I}_1^c}(x_2)$, and thereby obtain an integral-based approximate to
$\mathcal{P}_{0}^{p}$ in the following proposition.
\begin{proposition}\label{proposition: prob_suc_stage2}
The successful transmission probability by  the proposed scheme  for the case of
$0<\gamma_1<\beta$ is approximated as
\begin{align}
 \!\!\!\mathcal{P}_{0}^{p} \!
\approx \!\! \int_{0}^{\infty}\!\!\! e^{-h_{00}} \!\!
\int_{0}^{\frac{h_{00}}{\beta d^{\alpha}}} \!\!\!\! f_{\hat{I}_1}(x_1) \! \!
\int_{0}^{\frac{h_{00}}{\gamma_1 d^{\alpha}}-x_1} \!\!\!\!\!
f_{\hat{I}_1^c}(x_2)\,dx_2 \,dx_1 \,dh_{00}. \label{eq: prob_suc_stage2}
\end{align}
\end{proposition}
\begin{proof}
Please refer to Appendix B.
\end{proof}

Finally, by multiplying $\lambda_0$ with the right-hand side of (\ref{eq:
prob_suc_stage2}), we obtain an integral-based approximate to  $\mathcal{C}^{p}$
for the case of $0<\gamma_1<\beta$ as
\begin{equation}
 \mathcal{C}^{p} \approx \lambda_0 \int_{0}^{\infty}\!\!\! e^{-h_{00}} \!\!
\int_{0}^{\frac{h_{00}}{\beta d^{\alpha}}} \!\!\!\! f_{\hat{I}_1}(x_1) \! \!
\int_{0}^{\frac{h_{00}}{\gamma_1 d^{\alpha}}-x_1} \!\!\!\!\!
f_{\hat{I}_1^c}(x_2)\,dx_2 \,dx_1 \,dh_{00}. \label{eq: c_p_proposed}
\end{equation}

Note that  the proposed approximate approach considers the correlation between
$\textrm{SIR}_0^{(0)}$ and $\textrm{SIR}_0^{(1)}$, and only adopts PPP-based
approximation to approximate   $\Phi_1$ and $\Phi_1^{c}$ by $\hat{\Phi}_1$ and
$\hat{\Phi}_1^{c}$, respectively. Since it has been shown in the literature
(e.g., \cite{Hasan.TWC.2007.GuardZone}-\cite{Baccelli.Computer.11}) that such
PPP-based approximation can provide tight approximate to the corresponding
non-PPP, the proposed approximate approach is able to provide tight spatial
capacity approximate to $\mathcal{C}^{p}$ for the case of $0<\gamma_1<\beta$.

\subsubsection{Closed-form Approximation for (\ref{eq: c_p_proposed})}
Although the spatial capacity expression obtained in (\ref{eq: c_p_proposed}) is
easy to integrate, it is not of closed-form.
Thus, based on (\ref{eq: prob_suc_stage2}), we focus  on finding a closed-form
approximate to  $\mathcal{P}_{0}^{p}$ and thus $\mathcal{C}^{p}$.
We first increase the upper limit of $f_{\hat{I}_1^c}(x_2)$ in (\ref{eq:
prob_suc_stage2}) from $\gamma_1 d^{\alpha}-x_1$ to  $\gamma_1 d^{\alpha}$ to
obtain an upper bound for the right-hand side of (\ref{eq: prob_suc_stage2}).
Then by  properly lower-bounding the obtained upper bound based on Chebyshev's
inequality \cite{Chebyshev}, we obtain a closed-form approximate to
$\mathcal{P}_{0}^{p}$, which is shown in the following proposition.
\begin{proposition} \label{proposition: method_2}
Based on the integral-based expression given in  (\ref{eq: prob_suc_stage2}), a
closed-form approximate to $\mathcal{P}_{0}^{p}$ for the case of
$0<\gamma<\beta$ is obtained as
\begin{equation}
 \mathcal{P}_{0}^{p} \approx \exp(-\pi \lambda_1 d^2
\beta^{\frac{2}{\alpha}}\rho) \exp(-\pi \lambda_1^c d^2
\gamma_1^{\frac{2}{\alpha}}\rho). \label{eq: tractable_suc_prob}
\end{equation}
\end{proposition}
\begin{proof}
Please refer to Appendix C.
\end{proof}

From (\ref{eq: C_p_1}), (\ref{eq: lambda_1_expression}) and (\ref{eq:
tractable_suc_prob}), we  obtain a closed-form approximate to spatial capacity
of the proposed scheme for the case of $0<\gamma_1 < \beta$ as
\begin{align}
\mathcal{C}^{p}  \approx & \lambda_0 \times \exp\big(-\pi \lambda_0 d^2
\gamma_1^{\frac{2}{\alpha}}\rho\big) \times \exp\big[-\pi \lambda_0\exp(-\pi
\lambda_0 d^2 \gamma_1^{\frac{2}{\alpha}}\rho) d^2
\beta^{\frac{2}{\alpha}}\rho\big] \nonumber \\
&\times \exp\big[ \pi \lambda_0 \exp(-\pi \lambda_0 d^2
\gamma_1^{\frac{2}{\alpha}}\rho) d^2 \gamma_1^{\frac{2}{\alpha}}\rho\big].
\label{eq: tractable_sc}
\end{align}

\subsubsection{Conventional  Approximation}
It is noted that the conventional approximate approach in the literature (e.g.,
\cite{Hasan.TWC.2007.GuardZone}-\cite{Baccelli.Computer.11}), which only focuses
on  dealing with  the non-PPP $\Phi_1$,  can often yield a closed-form
expression. Thus, in the following, we apply the  conventional approximate
approach and discuss its approximate performance to  $\mathcal{C}^{p}$.

First, since only the performance in $\Phi_1$ is concerned by the conventional
approximate approach, it   takes $\mathbb{P}\big(\textrm{SIR}_0^{(1)}\geq
\beta\big)$  as the successful transmission probability of the typical
transmitter in the D-Phase.
Next,  the non-PPP $\Phi_1$ is approximated by the  homogeneous PPP
$\hat{\Phi}_1$ under Assumption 1.
We denote the received SIR at the typical receiver in  $\hat{\Phi}_1$ as
$\textrm{SIR}_0^{(\hat{1})}=h_{00}d^{-\alpha}/\sum_{i\in \hat{\Phi}_1}
h_{i0}l_{i0}$. Thus, $\mathbb{P}\big(\textrm{SIR}_0^{(1)}\geq \beta\big)$  is
approximated by
$\mathbb{P}\big(\textrm{SIR}_0^{(\hat{1})}\geq \beta\big)$.  At last, by
adopting  the product of $\lambda_1$ and
$\mathbb{P}\big(\textrm{SIR}_0^{(\hat{1})}\geq \beta\big)$ as an approximate to
the spatial capacity $\mathcal{C}^{p}$,  a  closed-form approximate to
$\mathcal{C}^{p}$ for the case of $0<\gamma_1 < \beta$ is obtained as
\begin{align}
\mathcal{C}^{p} \approx& \lambda_1 \times
\mathbb{P}\big(\textrm{SIR}_0^{(\hat{1})}\geq \beta\big) \label{eq:
conventional} \\
\overset{(a)}{=}& \lambda_0 \exp\big(-\pi \lambda_0 d^2
\gamma_1^{\frac{2}{\alpha}}\rho\big) \exp\Big[-\pi \lambda_0 \exp\big(-\pi
\lambda_0 d^2 \gamma_1^{\frac{2}{\alpha}}\rho\big) d^2
\beta^{\frac{2}{\alpha}}\rho\Big] \label{eq: conventional_result}
\end{align}
where $(a)$ follows by  Proposition \ref{proposition: SIR_1}  and (\ref{eq:
lambda_1_expression}).
Note that since $\lambda_1=\lambda_0 \times
\mathbb{P}\big(\textrm{SIR}_0^{(0)}\geq \gamma_1\big)$, we can rewrite (\ref{eq:
conventional}) as $\mathcal{C}^{p} \approx  \lambda_0 \times
\mathbb{P}\big(\textrm{SIR}_0^{(0)}\geq
\gamma_1\big)\mathbb{P}\big(\textrm{SIR}_0^{(\hat{1})}\geq \beta\big)$ under the
conventional method.
However, according to the definition of $\mathcal{C}^{p}$ for $N=1$,  which is
given in (\ref{eq: P_succ_p_1}) and (\ref{eq: C_p_1}), we have
$\mathcal{C}^{p} =  \lambda_0 \times \mathbb{P}\big(\textrm{SIR}_0^{(0)}\geq
\gamma_1,\textrm{SIR}_0^{(1)}\geq \beta\big)$, where the distribution of
$\textrm{SIR}_0^{(1)}$ is strongly dependent on that of  $\textrm{SIR}_0^{(0)}$
as $\Phi_1 \subseteq \Phi_0$. As a result,    the conventional approximate
approach only focuses on the PPP-based approximate  to $\Phi_1$, but ignores the
dependence between  $\Phi_0$ and $\Phi_1$. Therefore, (\ref{eq: conventional})
does not hold for representing, or reasonably approximating, the spatial
capacity of the proposed scheme. In addition, by comparing (\ref{eq:
tractable_sc}) and (\ref{eq: conventional_result}), it is observed that for the
case of $0<\gamma_1<\beta$, given any $\lambda_0>0$ and $d>0$, the closed-form
spatial capacity obtained  based on the proposed approach is always outperformed
that by the conventional approach.

\subsection{Numerical Results} \label{section: simulation_1}
Numerical results are presented in this subsection. According to the method
described in \cite{Stoyan.SG.95}, we generate a spatial Poisson process, in
which the transmitters are placed uniformly in a square of
$[0\textrm{m},600\textrm{m}]\!\times\![0\textrm{m},600\textrm{m}]$. To take care
of the border effects, we focus on sampling the transmitters that locate in the
interim square of
$[200\textrm{m},400\textrm{m}]\!\times\![200\textrm{m},400\textrm{m}]$.
 We calculate the spatial capacity as the average of the
network capacity over 2000 independent network realizations, where for each network realization,
the network capacity is  evaluated as the ratio of the number of successful
transmitters in the sampling square to the square area of
$4\!\times\!10^2\textrm{m}^2$.
Unless otherwise specified, in this subsection, we set $\alpha\!=\!4$,
$\beta=2.5$, and $d\!=\!10$m. We also observe by simulation that similar
performance can be obtained by using other parameters.

In the following, we first validate our analytical results on the spatial capacity of the proposed scheme and
the reference scheme without scheduling.  To highlight the spatial capacity improvement
performance of the proposed scheme, we then compare the spatial capacity achieved by the propose scheme with that by two existing
distributed scheduling schemes: one is the probability-based scheduling in \cite{Kim.ProbScheld.14}, and the other is the channel-threshold based scheduling in \cite{Weber.IT.07} and \cite{FadingScheduling}. At last, we consider a more practical scenario with SIR estimation and feedback errors, and show the
effects of the SIR errors on the spatial capacity of the proposed scheme.

\subsubsection{Validation of the Spatial Capacity Analysis}
We validate our spatial capacity analysis in Section \ref{section: exact_analysis_single} and Section \ref{section: spatial_capacity_characterization}
for both proposed and reference schemes.

Fig.~\ref{fig:3} shows  the spatial capacity versus the SIR threshold
$\gamma_1$, for both the reference  scheme without transmission scheduling and
the proposed  scheme with  SIR-based scheduling. We set the initial transmitter
density as $\lambda_0\!=\!0.0025/\textrm{m}^2$   in both  schemes. The
analytical spatial capacity of the reference scheme is given in (\ref{eq: C_r}).
By comparing the simulation results for the proposed scheme with the  analytical
results for the reference scheme,  we observe that $\mathcal{C}^r$ is  constant
over $\gamma_1$ as expected. We also observe that 1) when $\gamma_1\!<\!\beta$,
$\mathcal{C}^p\!>\!\mathcal{C}^r$; 2) when $\gamma_1\!=\!0$ or
$\gamma_1\!=\!\beta$, $\mathcal{C}^p\!=\!\mathcal{C}^r$; and 3) when
$\gamma_1\!>\!\beta$, $\mathcal{C}^p\!<\!\mathcal{C}^r$. This is in accordance
with our analytical results in Proposition \ref{proposition:
relationship_C1_C2}.
Moreover, for the proposed scheme,  we adopt (\ref{eq: C_r}) and (\ref{eq:
SC_gamma_higher_beta}) as the analytical spatial capacity  for the cases of
$\gamma_1=0$ and $\gamma_1\geq \beta$, respectively, and observe that the
analytical results of the spatial capacity fit well to the simulation
counterparts.
Furthermore, for the case of $0<\gamma_1<\beta$ of the proposed scheme, where
only  approximate expressions for the spatial capacity  are available, we
compare the approximate performance of the three approximate approaches given in
Section \ref{section: spatial_capacity_characterization}. It is observed that
the integral-based expression by the proposed approximate approach, given in
(\ref{eq: c_p_proposed}),  provides a tight approximate to  $\mathcal{C}^p$ for
the case of $0<\gamma_1<\beta$. In addition, as a cost of expressing in
closed-form,   (\ref{eq: tractable_sc}) is not as tight as (\ref{eq:
c_p_proposed}), but (\ref{eq: tractable_sc}) still provides a close approximate
to $\mathcal{C}^p$ for the case of $0<\gamma_1<\beta$.  At last,  it is observed
that the closed-form expression given in (\ref{eq: conventional_result}) by the
conventional approximate approach cannot properly approximate $\mathcal{C}^p$
for the case of $0<\gamma_1<\beta$ as expected.

\begin{figure}[t]
\centering
\DeclareGraphicsExtensions{.eps,.mps,.pdf,.jpg,.png}
\DeclareGraphicsRule{*}{eps}{*}{}
\includegraphics[angle=0, width=0.7\textwidth]{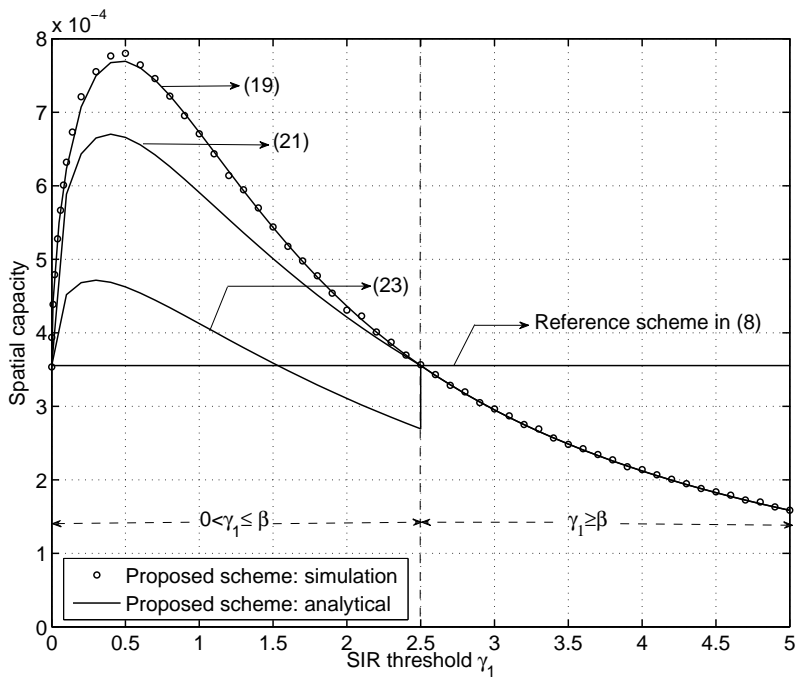}
\caption{Spatial capacity against $\gamma_1$ with $\lambda_0=0.0025$ and
$\beta=2.5$.}
\label{fig:3}
\end{figure}

\begin{figure}
\centering
\DeclareGraphicsExtensions{.eps,.mps,.pdf,.jpg,.png}
\DeclareGraphicsRule{*}{eps}{*}{}
\includegraphics[angle=0, width=0.7\textwidth]{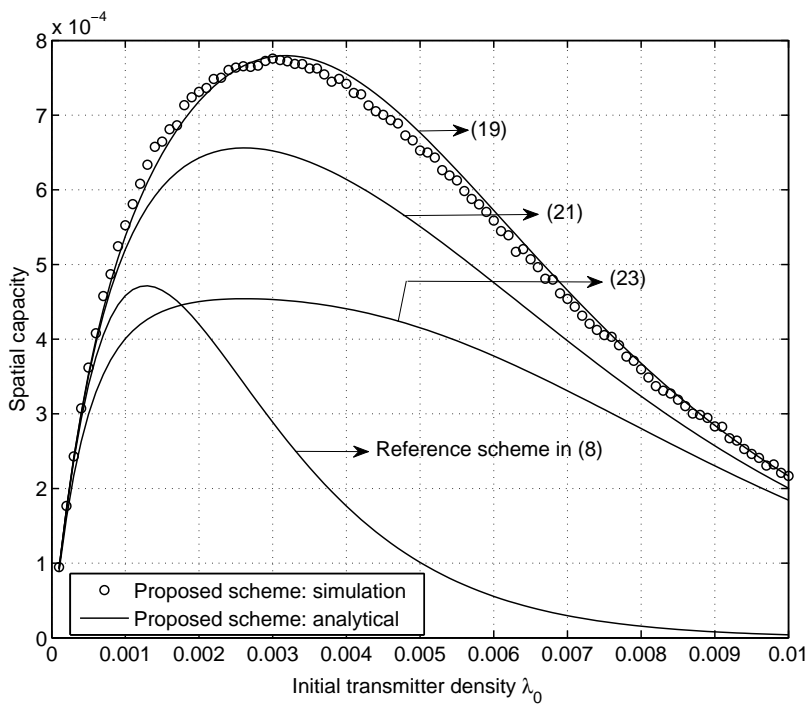}
\caption{Spatial capacity against $\lambda_0$. $\gamma_1=0.6$. $\beta=2.5$.}
\label{fig:4}
\end{figure}

Fig.~\ref{fig:4} shows the spatial capacity versus the initial transmitter
density $\lambda_0$ when $\gamma_1<\beta$. We set  $\gamma_1\!=\!0.6$. For the
proposed scheme, similar to the case in Fig.~\ref{fig:3}, we observe tight and
close approximates are provided by (\ref{eq: c_p_proposed}) and (\ref{eq:
tractable_sc}), respectively, based on the proposed approximate approach,  while
improper approximate is provided by (\ref{eq: conventional_result}) based on the
conventional approximate approach.  Moreover, it is observed that the spatial
capacity of the proposed scheme is always larger than that of the reference
scheme, given in (\ref{eq: C_r}), for all values of $\lambda_0$, which is as
expected from Proposition \ref{proposition: relationship_C1_C2} since
$\gamma_1\!<\!\beta$ in this example.
Furthermore, for both the proposed and reference schemes, we observe an interesting
\emph{density-capacity tradeoff}: by increasing  $\lambda_0$, the spatial
capacity first increases  due to more available transmitters, but as $\lambda_0$
exceeds a certain threshold, it starts to decrease, due to the more dominant
interference effect. Thus,  to maximize the spatial capacity,  under the system scenario set in Fig.~\ref{fig:4},  the optimal $\lambda_0$ should be set as $0.003/m^2$.

\subsubsection{Performance Comparison with Existing Distributed Schemes}
We consider two existing distributed scheduling schemes for performance comparison.
The first scheme   is the iterative probability-based scheduling as in \cite{Kim.ProbScheld.14}.
Denote the transmission probability for transmitter $i$ in P-Phase $k$, $1\leq k\leq N-1$,
and the D-Phase as $\phi_i^{(k)}$ or  $\phi_i^{(N)}$, respectively.
For any $k\in\{1,..,N\}$,  \cite{Kim.ProbScheld.14} sets $\phi_i^{(k)}=\min \left(\frac{\textrm{SIR}_i^{(k-1)}}{\beta},1 \right)$.
Intuitively, \cite{Kim.ProbScheld.14} provides a simple and proper way to iteratively adjust
the transmission probability $\phi_i^{(k)}$.
The second scheme  is the channel-threshold  based scheduling with single-stage probing as in \cite{Weber.IT.07}
and \cite{FadingScheduling}, where the received  interference power is
not involved in the transmission decision and each transmitter decides to
transmit in the D-Phase if its direct channel  strength in P-Phase $0$ is no
smaller than a predefined threshold $\gamma_1'$, i.e., $h_{ii}\!\geq
\!\gamma_1'$.
For a fair comparison, we consider single-stage probing with $N=1$ for all the proposed SIR-threshold based scheme,
the probability-based scheduling in \cite{Kim.ProbScheld.14}, and the channel-threshold based
scheduling in \cite{Weber.IT.07} and \cite{FadingScheduling}. 

\begin{figure}
\centering
\DeclareGraphicsExtensions{.eps,.mps,.pdf,.jpg,.png}
\DeclareGraphicsRule{*}{eps}{*}{}
\includegraphics[angle=0, width=0.7\textwidth]{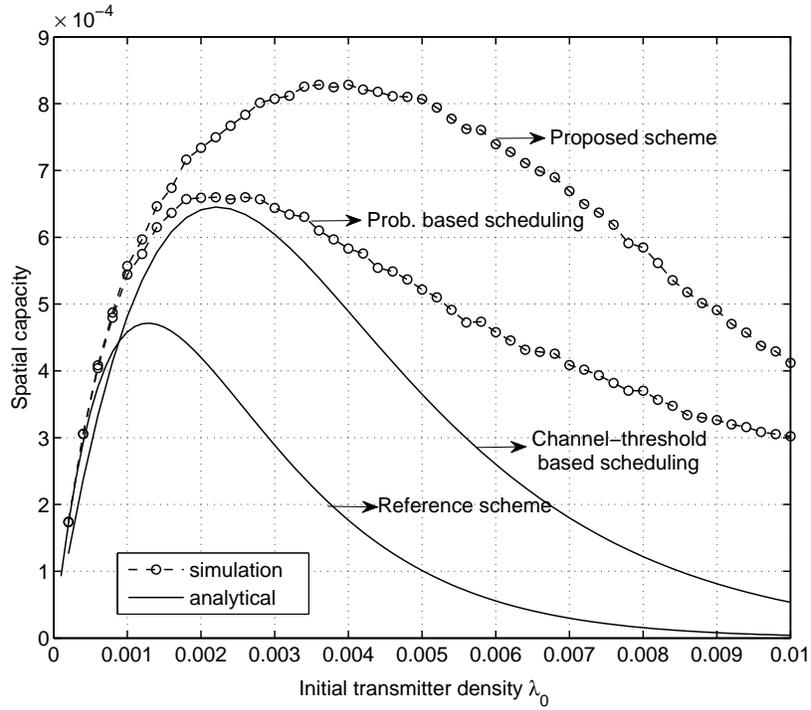}
\caption{Spatial capacity comparison with existing distributed scheduling schemes. $\gamma_1=0.4$. $\beta=2.5$.}
\label{fig:5}
\end{figure}

Fig.  \ref{fig:5} shows the spatial capacities achieved by the proposed scheme, the probability-based scheduling,
 the channel-threshold based scheduling, and the reference scheme without scheduling.
To clearly show the effects of involving
interference in the  transmission decision for the proposed scheme, we set
$\gamma_1'\!=\!\gamma_1=0.4$ for the channel-threshold based scheduling.
We obtain the spatial capacity of the channel-threshold based scheduling
by applying its exact expression given in \cite{FadingScheduling}. Due to the lack of an
exact spatial capacity expression for the probability-based scheduling, we obtain its  spatial capacity
by simulation. We list our observations from Fig.~\ref{fig:5}   as follows:
\begin{itemize}
 \item \emph{SIR based schemes v.s. channel-threshold based scheme:} It is observed that by adapting the transmission decision
to the SIR, the achieved spatial capacities by both the proposed scheme and the probability-based scheduling
are always higher than that by the channel-threshold based scheduling, where the interference information is not
exploited. Moreover, the spatial capacity of the channel-threshold based scheduling
is smaller than that of the reference scheme when $\lambda_0$ is small, and
becomes larger when $\lambda_0$ is sufficiently large. This is in sharp contrast
to the cases of the proposed scheme and the probability-based scheduling, which always guarantee  capacity improvement
over the reference scheme without scheduling.
\item \emph{SIR-threshold based scheduling v.s. probability-based scheduling:} It is interesting to observe that although both the proposed
scheme and the probability-based scheduling adapt the transmission decision to the SIR,
the achieved spatial capacity by the former scheme is always higher than that by the latter one  in this simulation.
This is because that the proposed scheme assures the improvement of the successful transmission probability of each
retained transmitter in the D-Phase, while the probability-based scheduling only assures such improvement with
some probability. Moreover, it is  observed that the optimal initial transmitter density that maximizes the spatial capacity of the proposed scheme is
$\lambda_0^{*}=0.0036$, which is larger than that for the probability-based scheduling locating at
$\lambda_0^{*}=0.0026$.
\end{itemize}
Note that  for the proposed scheme,  a lower SIR threshold $\gamma_1$ allows more transmitters to retain in the D-Phase, so as to have a second chance to transmit. Thus,  by comparing  the simulation results of the proposed scheme in Fig.~\ref{fig:5} with that in Fig.~\ref{fig:4}, it is observed that the achieved optimal spatial capacity over $\lambda_0$ with $\gamma_1=0.4$  in Fig.~\ref{fig:5} is larger than that with $\gamma_1=0.6$  in Fig.~\ref{fig:4}.
In addition, for all the considered schemes in Fig. \ref{fig:5}, we observe a density-capacity tradeoff, which is similar to that in Fig.~\ref{fig:4}.

\subsubsection{Effects of the SIR Estimation and Feedback Errors}
We  consider a more practical scenario, where SIR estimation and feedback errors exist in
the implementation of the proposed scheme, and show the effects of the SIR errors on
the spatial capacity.  Similarly to \cite{fed.noise}, where the channel estimation and feedback errors are assumed 
to be zero-mean Gaussian variables, respectively, we assume  the SIR estimation and feedback errors  follow zero-mean Gaussian distributions with variance $\sigma^2_{est}$ and $\sigma^2_{fed}$, respectively. By further assuming that the two types of
SIR errors are mutually independent, the sum  of both SIR errors   at  transmitter $i$, denoted by $n_i$, follows
zero-mean Gaussian distribution with variance $\sigma^2=\sigma^2_{est}+\sigma^2_{fed}$.
Thus, in the presence of SIR errors, the feedback SIR level at transmitter $i$ in P-Phase 0 is $\textrm{SIR}_i^{(0)}+n_i$.
Moreover, if the feedback SIR level $\textrm{SIR}_i^{(0)}+n_i\geq \gamma_1$ for a given SIR threshold $\gamma_1$ in P-Phase 1,
transmitter $i$  decides to transmit in P-Phase $1$; otherwise, it decides to be idle in the remaining time
of this time slot.
Similar to its counterpart without SIR errors in Fig. \ref{fig:4} and Fig. \ref{fig:5}, the spatial capacity with SIR errors is calculated as an average value
over all the transmitters' random locations, the random fading channels, as well as the random SIR errors.
\begin{figure}
\centering
\DeclareGraphicsExtensions{.eps,.mps,.pdf,.jpg,.png}
\DeclareGraphicsRule{*}{eps}{*}{}
\includegraphics[angle=0, width=0.7\textwidth]{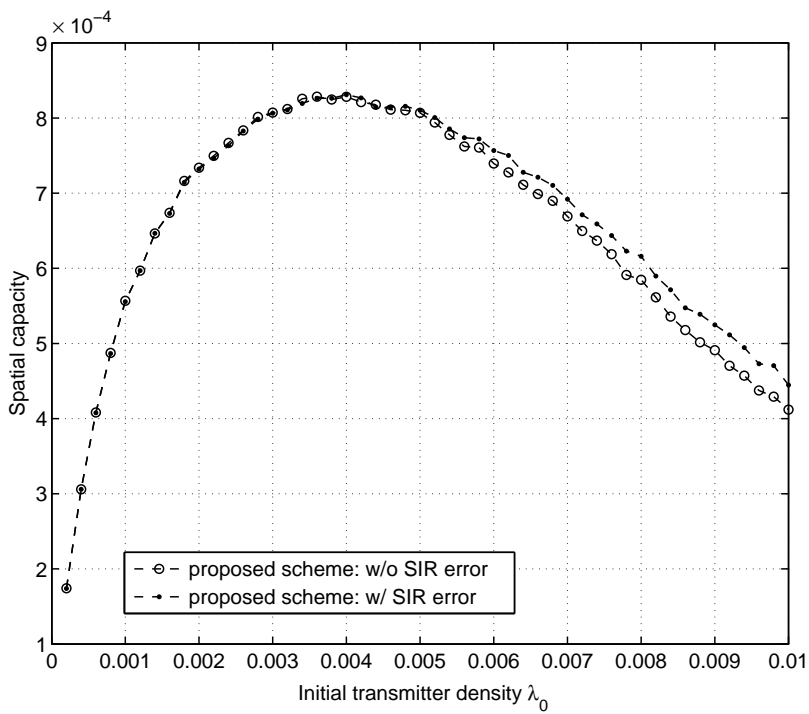}
\caption{Effects of the SIR errors on the spatial capacity of the proposed scheme. $\gamma_1=0.4$. $\beta=2.5$.}
\label{fig:e}
\end{figure}

Fig. \ref{fig:e}  numerically shows  the spatial capacities of the proposed scheme
in both cases with and without SIR errors.
We set $\sigma^2=10^{-2}$ and $\gamma_1=0.4$ in this example.
 It is observed from Fig. \ref{fig:e} that when $\lambda_0$ is small, due to the resultant small interference in the network, each receiver $i$ feeds back a
sufficiently high SIR level $\textrm{SIR}_i^{(0)}$ to its associated transmitter, such that $n_i$ has a small probability to affect the transmitter's decision.
Thus, we observe that when $\lambda_0$ is small, the spatial capacity with SIR errors is tight to that without SIR errors. However, as $\lambda_0$ increases,
due to the decreased   $\textrm{SIR}_i^{(0)}$ at each transmitter $i$,   the transmitters become
more easily  affected by the SIR errors $n_i$ when deciding whether to transmit based on $\textrm{SIR}_i^{(0)}+n_i\geq \gamma_1$.
It is noted that when $\lambda_0$ is sufficiently large, the average SIR level at each transmitter becomes very small; and even if
 $\textrm{SIR}_i^{(0)}\geq \gamma_1$ for transmitter $i$, $\textrm{SIR}_i^{(0)}$ is close to $\gamma_1$ with a large probability.
 Thus, under the case with zero-mean Gaussian distributed error $n_i$, for the transmitters  with
$\textrm{SIR}_i^{(0)}\geq \gamma_1$ in the SIR error-free case, it is more likely that these transmitters  become $\textrm{SIR}_i^{(0)}+n_i< \gamma_1$ than
$\textrm{SIR}_i^{(0)}+n_i\geq  \gamma_1$ in the SIR error-involved case. Similarly, we can
easily find that for the transmitters with $\textrm{SIR}_i^{(0)}< \gamma_1$ in the SIR error-free case,  it is also more likely that these transmitters maintain
$\textrm{SIR}_i^{(0)}+n_i< \gamma_1$ than $\textrm{SIR}_i^{(0)}+n_i\geq  \gamma_1$ in the SIR error-involved case.
Thus, the number of transmitters  with $\textrm{SIR}_i^{(0)}+n_i< \gamma_1$ in the SIR error-involved case is larger than that
with $\textrm{SIR}_i^{(0)}< \gamma_1$ in the SIR error-free case in general.
Hence, as compared to the case without SIR errors, more transmitters will be refrained from transmitting in the D-Phase in the case
with SIR errors, which improves the successful transmission probability in the D-Phase due to the reduced interference.
As a result,  it is interesting to observe from Fig. \ref{fig:e} that
  when the initial transmitter density $\lambda_0$  increases to some significant point, the spatial capacity
with SIR errors becomes slightly higher than that without SIR errors; and their gap slowly increases over $\lambda_0$ after this point.
Therefore, inaccurate SIR may even help improve the SIR-based scheduling performance in more interference-limited regime, which makes the proposed design robust to SIR errors.

\section{SIR-Threshold based Scheme with Multi-Stage Probing} \label{section:
probing}
In this section, we consider the proposed scheme with multi-stage probing, i.e.,
$N>1$.
In this case, $N$ probing phases are sequentially implemented  to gradually
decide the transmitters that are allowed to transmit in the data transmission phase. 
According to (\ref{eq: C_p_def}), to find the spatial capacity $\mathcal{C}^{p,N}$ with  $N$ probing phases,
we need to first find the successful transmission probability $\mathcal{P}_0^{p,N}$ given in (\ref{eq: P_succ_p_def}).
However, due to the mutually coupled user transmissions over different probing phases,
the successful transmission probability in P-Phase $k$, $0<k\leq N$, is related to the SIR distributions in all the proceeding probing phases (from P-Phase 0 to P-Phase $k-1$).
Moreover, due to the different point process formed by the retained transmitters in each probing phase, the SIR correlations of any two
probing phases are  different.
Thus, it is challenging to express the successful transmission probability and thus the spatial capacity for the case with $N>1$ in general.
As a result, instead of focusing on expressing the spatial capacity $\mathcal{C}^{p,N}$, we focus on
studying how the key system design parameters, such as the SIR thresholds and the number of probing phases $N$, affect
the spatial capacity of the proposed scheme with $N>1$. 
In particular, unlike the case  with $N=1$, where the single-stage overhead $N\tau=\tau \ll T$
is negligible, the multi-stage overhead $N \tau$ with $N>1$ may not be
negligible.
In the following, we first study the impact of multiple SIR thresholds on the
spatial capacity by extending Proposition \ref{proposition: relationship_C1_C2}
for the case of $N=1$ to the case of $N>1$. We then investigate the effects of
the multi-stage probing overhead  on the spatial capacity.

\subsection{Impact of SIR Thresholds}
From (\ref{eq: P_succ_p_def}) and (\ref{eq: C_p_def}), the spatial capacity of
the proposed scheme is determined by the values of  SIR thresholds as well as
the time overhead $N\tau$ for probing. To focus on the impact of the SIR
thresholds, in this subsection, we  assume   $N\tau$  is negligible and thus
have
\begin{equation}
 \mathcal{C}^{p,N}= \lambda_0 \mathbb{P}\big(\textrm{SIR}_0^{(0)}\geq
\gamma_1,...,\textrm{SIR}_0^{(N-1)}\geq \gamma_N, \textrm{SIR}_0^{(N)}\geq
\beta\big) \label{eq: C_p_N}
\end{equation}
 where  the distributions of  $\textrm{SIR}_0^{(k)}$'s, $0\leq k \leq N$,  are
mutually dependent and all the $\Phi_k$'s,  $1\leq k \leq N$, are non-PPPs in
general. It is also noted that for any $1\leq k \leq N$, we have $\Phi_k
\subseteq \Phi_{k-1}$ for $\gamma_k\geq0$. Thus, the network interference level
in $\Phi_k$ is reduced, as compared to that in  $\Phi_{k-1}$. As a result, by
extending Proposition \ref{proposition: relationship_C1_C2} for the case of
$N=1$, we obtain the following proposition for the case of $N>1$.

\begin{proposition} \label{proposition: relationship_cn_cn-1}
Consider two proposed schemes with arbitrary $N-1$ and $N$ probing phases,
respectively,  $N> 1$.
Suppose  the two schemes adopt the same SIR threshold  $\gamma_k\geq0$ in each
$\Phi_k$, $\forall k\in\{1,...,N-1\}$.
Then given $\beta>0$,  by varying the SIR threshold $\gamma_{N}\in [0,\infty)$
in the data transmission phase for the proposed scheme with $N$ probing phases,
 we  have the following relationship between $\mathcal{C}^{p,N}$ and
$\mathcal{C}^{p,N-1}$ based on (\ref{eq: C_p_N}):
\begin{equation}
  \left\{
   \begin{array}{l}
    \mathcal{C}^{p,N} > \mathcal{C}^{p,N-1},~\text{if}~\gamma_{N-1
}<\gamma_{N}<\beta \textrm{~\big(\emph{conservative} transmission regime\big)}\\
    \mathcal{C}^{p,N}=\mathcal{C}^{p,N-1},~\text{if}~0\leq \gamma_{N}\leq
\gamma_{N-1}~\text{or}~\gamma_{N}=\beta \textrm{~\big(\emph{neutral}
transmission regime\big)}\\
    \mathcal{C}^{p,N} < \mathcal{C}^{p,N-1},~\text{if}~\gamma_{N}>\beta
\textrm{~\big(\emph{aggressive} transmission regime\big)}.
   \end{array}
  \right.   \label{C_p_N-1_N}
\end{equation}
\end{proposition}

\begin{proof}
Please refer to Appendix D.
\end{proof}

\begin{remark}
Similar to the case of Proposition \ref{proposition: relationship_C1_C2}, in
Proposition  \ref{proposition: relationship_cn_cn-1}, in the conservative
transmission regime with  $\gamma_{N-1 }<\gamma_{N}<\beta$,  we obtain improved
spatial capacity; in  the aggressive transmission regime   with
$\gamma_{N}>\beta$, we obtain reduced  spatial capacity; and in the neutral
transmission region with $0\leq \gamma_{N}\leq \gamma_{N-1}$ or
$\gamma_{N}=\beta$, we obtain unchanged capacity. Moreover, based on the fact
that the conservative transmission decision is beneficial for improving the
spatial capacity of the proposed scheme, we obtain the following corollary,
which gives a proper method to  set  the values of all the SIR-thresholds, such
that the improvement of spatial capacity over the number of probing phases is
assured.
\end{remark}

\begin{corollary} \label{corollary: N>1}
For a proposed scheme with $N\!>\!1$ probing phases with  negligible overhead,
if the designed SIR thresholds are properly increased as
$0<\gamma_1<\cdots<\gamma_N<\beta$, the resulting spatial capacity
$\mathcal{C}^{p,N}$ increases with the  number of probing phases $N$.
\end{corollary}

It is worth noting that  based on (\ref{eq: C_p_N}), for a given $\lambda_0>0$,
$\mathcal{C}^{p,N}$ is only determined by the successful transmission
probability $\mathcal{P}_{0}^{p,N}$, given in (\ref{eq: P_succ_p_def}). Thus,
both Proposition  \ref{proposition: relationship_cn_cn-1} and Corollary
\ref{corollary: N>1} also apply for $\mathcal{P}_{0}^{p,N}$.

\begin{figure}
\centering
\DeclareGraphicsExtensions{.eps,.mps,.pdf,.jpg,.png}
\DeclareGraphicsRule{*}{eps}{*}{}
\includegraphics[angle=0, width=0.7\textwidth]{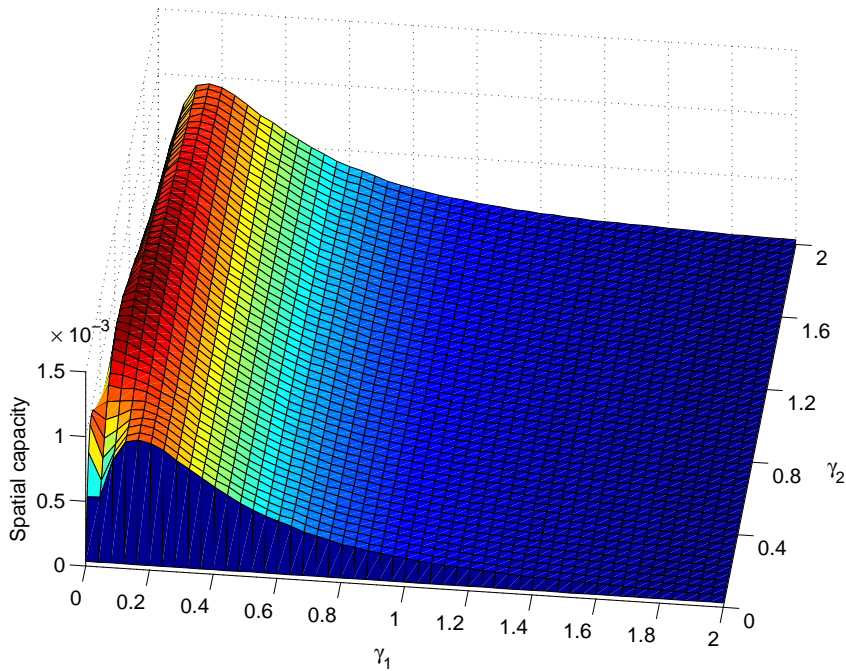}
\caption{Spatial capacity against $\gamma_1$ and $\gamma_2$ for $N=2$.}
\label{fig:gamma_1_2}
\end{figure}
In the next,  we provide a numerical example with $N=2$ to further discuss the
impact of SIR thresholds on the spatial capacity. In this example,  we set
$\alpha=4$, $\beta=2$, and $d=10$.  Fig.~\ref{fig:gamma_1_2} shows the
corresponding spatial capacity over $\gamma_1\in [0,\beta]$ and $\gamma_2\in
[0,\beta]$.
It is observed from Fig.~\ref{fig:gamma_1_2} that if $0 \leq \gamma_2\leq
\gamma_1$, the spatial capacity  achieved at $(\gamma_1, \gamma_2)$
remains unchanged over $\gamma_2$; and if $\gamma_2 > \gamma_1$, the spatial
capacity achieved at $(\gamma_1, \gamma_2)$ is always larger than that achieved
at $(\gamma_1, 0)$.
Apparently, this is in accordance with Proposition \ref{proposition:
relationship_cn_cn-1}. Moreover, among all the points over $\gamma_1\in
[0,\beta]$ and $\gamma_2\in [0,\beta]$, such trend is more obviously observed
for   small $\gamma_1$ and small $\gamma_2$.
In addition,  it is also observed that the spatial capacity varies much faster
over $\gamma_1$ than over $\gamma_2$, and when $\gamma_1$ is sufficiently large,
 the resulting spatial capacity does not change much over $\gamma_2$.   As a
result, the SIR threshold $\gamma_1$ plays a more critical role in determining
the spatial capacity than $\gamma_2$, since $\gamma_1$ determines how many
transmitters can have a second chance to contend the transmission opportunity.
Furthermore, it is observed that \emph{ to achieve a higher spatial capacity, it
is  preferred to start with a small $\gamma_1>0$, and then set $\gamma_2<\beta$
with an increasing step-size, i.e.,  $\gamma_1-0<\gamma_2-\gamma_1$.} As shown
in this example, the maximum spatial capacity is achieved at $\gamma_1=0.15$ and
$\gamma_2=0.8$; and the spatial capacity at $\gamma_1=0.15$ reduces very slowly
over $\gamma_2\in[0.8, \beta]$.

\subsection{Effects of Multi-Stage Overhead of Probing}
In this subsection, we assume the multi-stage overhead $N\tau$ for probing is
not negligible and study the effects of $N\tau$  on the spatial capacity.
In this case, it is easy to find from (\ref{eq: C_p_def}) that the effective
data transmission time is reduced over the probing-stage number $N$, which
reduces the spatial capacity. On the other hand, from  Corollary~\ref{corollary:
N>1}, under the constraint that $0<\gamma_1<\cdots<\gamma_N<\beta$, the
successful transmission probability increases over $N$.   As a result, from
(\ref{eq: C_p_def}), there exists a probing-capacity tradeoff over $N$ under the
condition that $0<\gamma_1<\cdots<\gamma_N<\beta$. In the following, we
illustrate the probing-capacity tradeoff by a  numerical example (see
Fig.~\ref{fig:time_overhead}).

\begin{figure}[t]
\centering
\DeclareGraphicsExtensions{.eps,.mps,.pdf,.jpg,.png}
\DeclareGraphicsRule{*}{eps}{*}{}
\includegraphics[angle=0, width=0.7\textwidth]{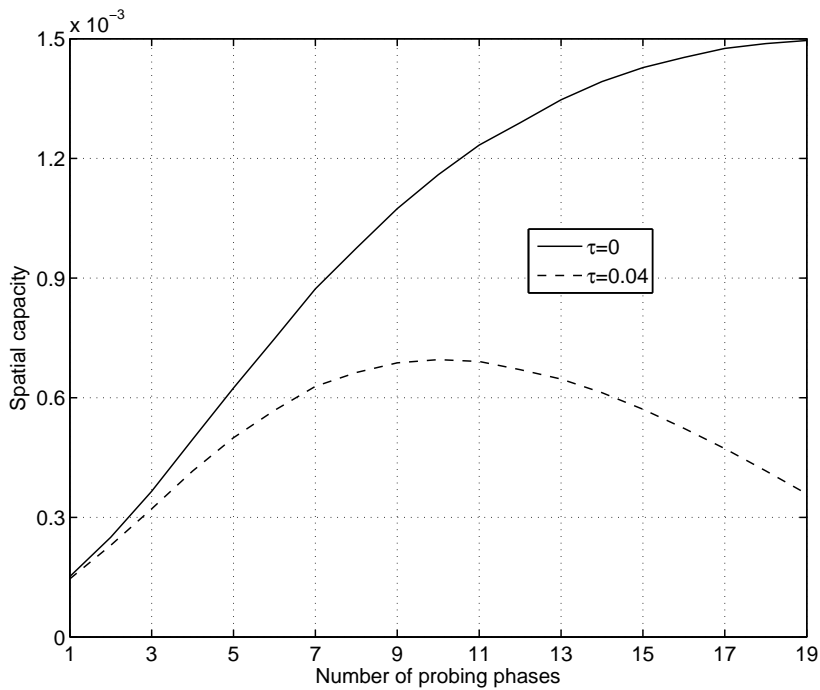}
\caption{Spatial capacity against the number of probing phases.}
\label{fig:time_overhead}
\end{figure}

In this example, we show the spatial capacity of the proposed scheme over the
number of probing phases $N$.  We set $\alpha=4$, $\beta=2$, $d=10$, and the
time slot duration $T=1$ second (s).  We consider two cases with $\tau=0$s and
$\tau=0.04$s,  respectively, where the time overhead for probing is zero for the
former case and non-negligible for the latter one. For both cases, as
enlightened by Fig.~\ref{fig:gamma_1_2},  we start with $\gamma_1$=0.01 and
gradually increase $\gamma_k$, $2\leq k \leq N$, based on
$\gamma_k=\gamma_{k-1}+0.01k$, which gives an increasing step-size with
$\gamma_k-\gamma_{k-1}<\gamma_{k+1}-\gamma_k$. To ensure  $\gamma_N<\beta$,  the
maximum allowable $N$ is obtained as $19$. As shown in
Fig.~\ref{fig:time_overhead}, it is observed that the spatial capacity increases
over $N$ for the case with   $\tau=0$, which  validates Corollary
\ref{corollary: N>1}. Moreover, for the case with  $\tau=0.04$, the
probing-capacity tradeoff is observed as expected:  \emph{the spatial capacity
first increases over $N$,
due to the improved performance of successful transmission probability,  but
after $N=10$, the spatial capacity begins to decrease over $N$, due to the more
dominant effects of the reduced  data transmission time.}

\section{Conclusion}
In this paper, we addressed the spatial capacity analysis and characterization
in a wireless ad hoc network by an efficient SIR-threshold based scheme.  For
single-stage probing, we showed the conditions under which the spatial capacity
of the proposed scheme performs strictly better than that of the reference
scheme without scheduling.
We also characterized the spatial capacity of the proposed scheme in
closed-form. In particular, we proposed a new approach to approximate the
spatial capacity,  which is useful for analyzing performance of wireless
networks with interacted  transmitters.
For  multi-stage probing, we extended the results for the case of single-stage
probing, and gave the condition  under which the spatial capacity of the
proposed scheme can be gradually improved over the probing-stage number. We also
studied the effects of multi-stage  probing overhead and investigated the
probing-capacity tradeoff.

 Although the considered on/off power control in this paper is more practical
 than the multi-level power control for implementation \cite{BPC.Prac}, it is interesting
 to extend our network-level performance analysis to the multi-level  power control in our future work.
One issue needs to be properly addressed is the power convergence in the stochastic network.
Unlike the power convergence studied in \cite{Foschini-Miljanic} and \cite{Yates}
for deterministic wireless networks,   the transmit power level of each transmitter in the current probing phase is stochastically 
determined by   the SIR distributions in all the proceeding  probing phases. Moreover, due to the different point process formed by the retained transmitter in each probing phase, the SIR distributions in all the probing phases are mutually different.
Although challenging, it is  of our interest to find the condition that assures the power convergence in our considered stochastic network, and study the spatial capacity in a stable system with converged power level of each transmitter. 
In addition, we are also interested to extend our current study on synchronized transmission  to the asynchronized transmission  in our future work.
Unlike the synchronized transmission, due to the newly added transmitters in each probing phase,
it is more difficult to control and analyze the interference in each probing phase. Moreover, note that the asynchronized transmission may
cause unstable communication quality for  the transmitters.
It is thus of our interest to design effective transmission scheme that  can assure   stable communication quality for all 
the transmitters, by effectively controlling  the network interference to improve the spatial capacity in our future work.

\appendices

\section{Proof of Proposition\ref{proposition:
relationship_C1_C2}}\label{appendix:proof 1}

By expressing $\textrm{SIR}_0^{(0)}=h_{00}d^{-\alpha}/I_0$ and
$\textrm{SIR}_0^{(1)}=h_{00}d^{-\alpha}/I_1$, based on (\ref{eq: P_succ_r_def}),
(\ref{eq: C_r_def}), (\ref{eq: P_succ_p_1_cond}) and (\ref{eq: C_p_1}), we have
\begin{equation}
\frac{\mathcal{C}^{p}}{\mathcal{C}^{r}}\!=\!\frac{\mathbb{P}(h_{00}\!\geq
\!\gamma_1 d^{\alpha}I_0) \!\times \! \mathbb{P}( h_{00}\!\geq \! \beta
d^{\alpha}I_1|h_{00}\! \geq \! \gamma_1 d^{\alpha}I_0)}{\mathbb{P}(h_{00}\geq
\beta d^{\alpha}I_0)}. \label{eq: CII_CI_ratio}
\end{equation}
In the following, we compare  $\frac{\mathcal{C}^{p}}{\mathcal{C}^{r}}$ with $1$
by varying $\gamma_1\in[0,\infty)$.
Clearly, when $\gamma_1=0$, $\frac{\mathcal{C}^{p}}{\mathcal{C}^{r}}=1$.
Next, we consider the case of $\gamma_1 \geq \beta$.
Since $I_0 \geq I_1$, if $\gamma_1 \geq \beta$, we obtain $\mathbb{P}(
h_{00}\geq \beta d^{\alpha}I_1|h_{00}\geq \gamma_1 d^{\alpha}I_0)=1$.
Moreover, for the non-negative and continuous random variables $h_{00}$ and
$I_0$, it is easy to find that if $\gamma_1 > \beta$,
 $\mathbb{P}(h_{00}\geq \gamma_1 d^{\alpha}I_0) < \mathbb{P}(h_{00}\geq \beta
d^{\alpha}I_0)$, and if $\gamma_1 = \beta$, $\mathbb{P}(h_{00}\geq \gamma_1
d^{\alpha}I_0) =\mathbb{P}(h_{00}\geq \beta d^{\alpha}I_0)$.
As a result, from (\ref{eq: CII_CI_ratio}), if $\gamma_1 > \beta$,
$\frac{\mathcal{C}^{p}}{\mathcal{C}^{r}}<1$, and   if $\gamma_1 = \beta$,
$\frac{\mathcal{C}^{p}}{\mathcal{C}^{r}}=1$.
At last, we consider the case of $0<\gamma_1<\beta$.  In this case, we have
$\mathbb{P}(h_{00}\geq \gamma_1 d^{\alpha}I_0|h_{00}\geq \beta
d^{\alpha}I_0)=1$, or equivalently,
\begin{align}
&\frac{\mathbb{P}(h_{00}\!\geq \! \gamma_1 d^{\alpha}I_0) \! \times \!
\mathbb{P}(h_{00}\! \geq \! \beta d^{\alpha}I_0|h_{00}\! \geq \! \gamma_1
d^{\alpha}I_0)}{\mathbb{P}(h_{00}\! \geq \! \beta d^{\alpha}I_0)}\! =\! 1.
\label{eq: SIR=1}
\end{align}
Moreover, since $\gamma_1\!\neq\!0$ in this case, we have
$F_{I_1}(x)\!>\!F_{I_0}(x)$, $\forall x\!>\!0$, where $F_{I_0}(\cdot)$ and
$F_{I_1}(\cdot)$ denote the cumulative distribution functions (CDFs) of $I_0$
and $I_1$, respectively. It is then easy to verify that
$\mathbb{P}(h_{00}\!\geq \! \beta d^{\alpha}I_1|h_{00}\! \geq \! \gamma_1
d^{\alpha}I_0)\! > \! \mathbb{P}(h_{00}\! \geq \! \beta d^{\alpha}I_0|h_{00}\!
\geq \! \gamma_1 d^{\alpha}I_0)$, for which, by multiplying
$\frac{\mathbb{P}(h_{00}\geq \gamma_1 d^{\alpha}I_0)}{\mathbb{P}(h_{00}\geq
\beta d^{\alpha}I_0)}$ on both sides  and based on (\ref{eq: SIR=1}), we have
$$
\frac{\mathbb{P}(h_{00}\geq \gamma_1 d^{\alpha}I_0) \times \mathbb{P}(
h_{00}\geq \beta d^{\alpha}I_1|h_{00} \geq\gamma_1
d^{\alpha}I_0)}{\mathbb{P}(h_{00}\geq \beta d^{\alpha}I_0)}> 1.
$$
That is, $\frac{\mathcal{C}^{p}}{\mathcal{C}^{r}}>1$.
Proposition \ref{proposition: relationship_C1_C2} thus follows.

\section{Proof of Proposition {\ref{proposition: prob_suc_stage2}}}
Under Assumption 1, we obtain two independent PPPs
$\hat{\Phi}_1$  and $\hat{\Phi}_1^{c}$, with $\hat{\Phi}_1\cup
\hat{\Phi}_1^{c}=\Phi_0$ and $\hat{\Phi}_1\cap \hat{\Phi}_1^{c}=\emptyset$.
Since from (\ref{eq: P_succ_p_1}), it follows that
\begin{align*}
\mathcal{P}_0^{p}=&\mathbb{P}(\textrm{SIR}_0^{(0)}\geq
\gamma_1,\textrm{SIR}_0^{(1)}\geq \beta) \\
=&\mathbb{P}\bigg(\sum_{i\in \Phi_0,i\neq 0}h_{i0}l_{i0}\leq
\frac{h_{00}}{\gamma_1 d^{\alpha}}, \sum_{i\in \Phi_1,i\neq 0}h_{i0}l_{i0}\leq
\frac{h_{00}}{\beta d^{\alpha}}\bigg),
\end{align*}
we have
\begin{align}
\mathcal{P}_0^{p}\approx & \mathbb{P}\bigg(\!\Big(\! \sum_{i\in
\hat{\Phi}_1,i\neq 0} \!h_{i0}l_{i0}\! + \! \sum_{i\in \hat{\Phi,i\neq
0}_1^{c}}h_{i0}l_{i0}\! \Big) \! \leq \! \frac{h_{00}}{\gamma_1 d^{\alpha}},
\sum_{i\in \hat{\Phi}_1}h_{i0}l_{i0}\! \leq \! \frac{h_{00}}{\beta d^{\alpha}}\!
\bigg) \nonumber \\
=&\mathbb{P}\Big(\hat{I}_1+\hat{I}_1^c\leq \frac{h_{00}}{\gamma_1 d^{\alpha}},
\hat{I}_1\leq \frac{h_{00}}{\beta d^{\alpha}}\Big). \nonumber
\end{align}
Due to the independence of $\hat{\Phi}_1$ and $\hat{\Phi}_1^{c}$, $\hat{I}_1$ is
independent of $\hat{I}_1^c$.
Given $h_{00}$, we thus have
\begin{align}
&\mathbb{P}\Big(\hat{I}_1+\hat{I}_1^c\leq \frac{h_{00}}{\gamma_1 d^{\alpha}},
\hat{I}_1\leq \frac{h_{00}}{\beta d^{\alpha}}\Big|h_{00}\Big)  \nonumber \\
=&\int_{0}^{\frac{h_{00}}{\beta d^{\alpha}}} f_{\hat{I}_1}(x_1)
\int_{0}^{\frac{h_{00}}{\gamma_1 d^{\alpha}}-x_1} f_{\hat{I}_1^c}(x_2)\,dx_2
\,dx_1. \label{eq: in_proof}
\end{align}
By integrating  (\ref{eq: in_proof}) over the (exponential) distribution of
$h_{00}$, we obtain  (\ref{eq: prob_suc_stage2}). Proposition {\ref{proposition:
prob_suc_stage2}} thus follows.

\section{Proof to Proposition \ref{proposition: method_2}}

In this proof, we  first derive an upper bound for the right-hand side of
(\ref{eq: prob_suc_stage2}), and then by properly lower-bounding the obtained
upper bound, we give a tractable approximate to $\mathcal{P}_{0}^{p}$.

First, in (\ref{eq: prob_suc_stage2}), by increasing the upper limit of
$f_{\hat{I}_1^c}(x_2)$, i.e., $\gamma_1 d^{\alpha}-x_1$, to  $\gamma_1
d^{\alpha}$,  the tight approximation of $\mathcal{P}_{0}^{p}$ is upper-bounded
as
\begin{align}
 \mathcal{P}_{0}^{p}
& \approx  \int_{0}^{\infty}\!\!\! e^{-h_{00}} \!\!
\int_{0}^{\frac{h_{00}}{\beta d^{\alpha}}} \!\!\!\! f_{\hat{I}_1}(x_1) \! \!
\int_{0}^{\frac{h_{00}}{\gamma_1 d^{\alpha}}-x_1} \!\!\!\!\!
f_{\hat{I}_1^c}(x_2)\,dx_2 \,dx_1 \,dh_{00} \nonumber \\
& < \int_{0}^{\infty}\!\!\! e^{-h_{00}} \!\! \int_{0}^{\frac{h_{00}}{\beta
d^{\alpha}}} \!\!\!\! f_{\hat{I}_1}(x_1) \! \! \int_{0}^{\frac{h_{00}}{\gamma_1
d^{\alpha}}} \!\!\!\!\! f_{\hat{I}_1^c}(x_2)\,dx_2 \,dx_1 \,dh_{00}.  \label{eq:
tight_upper}
\end{align}

Next, denote $Y_1(h_{00})=\int_{0}^{\frac{h_{00}}{\beta d^{\alpha}}} \!\!
f_{\hat{I}_1}(x_1) \,dx_1$ and $Y_2(h_{00})=\int_{0}^{\frac{h_{00}}{\gamma_1
d^{\alpha}}} \!\! f_{\hat{I}_1^c}(x_2) \,dx_2$. We can rewrite (\ref{eq:
tight_upper}) as
\begin{equation}
 \int_{0}^{\infty}\!\!\! e^{-h_{00}} \!\! \int_{0}^{\frac{h_{00}}{\beta
d^{\alpha}}} \!\!\!\! f_{\hat{I}_1}(x_1) \! \! \int_{0}^{\frac{h_{00}}{\gamma_1
d^{\alpha}}} \!\!\!\!\! f_{\hat{I}_1^c}(x_2)\,dx_2 \,dx_1 \,dh_{00} =
\mathbb{E}\big[Y_1(h_{00})Y_2(h_{00}) \big]. \label{eq: tight_upper_Expectation}
\end{equation}
Note that both $Y_1(h_{00})$ and $Y_2(h_{00})$ are monotonically increasing over
$h_{00}$. Thus,   according to the Chebyshev's inequality \cite{Chebyshev}, the
right-hand side of (\ref{eq: tight_upper_Expectation}) can be lower-bounded as
\begin{equation}
 \mathbb{E}\big[Y_1(h_{00})Y_2(h_{00}) \big] \geq
\mathbb{E}\big[Y_1(h_{00})\big] \mathbb{E}\big[Y_2(h_{00}) \big]. \label{eq:
tight_upper_lower}
\end{equation}
For $\mathbb{E}\big[Y_1(h_{00})\big]$ in (\ref{eq: tight_upper_lower}),  by
integrating $Y_1(h_{00})$ over the exponential distributed $h_{00}$, we obtain
that
\begin{align}
\mathbb{E}\big[Y_1(h_{00})\big]  =& \int_{0}^{\infty}\!\!\! e^{-h_{00}} \!\!
\int_{0}^{\frac{h_{00}}{\beta d^{\alpha}}} \!\!\  f_{\hat{I}_1}(x_1) \,dx_1
\,dh_{00} \nonumber \\
=& P \Big(0\leq \hat{I}_1 \leq \frac{h_{00}}{\beta d^{\alpha}}  \Big)  \nonumber
\\
\overset{(a)}{=}&\exp(-\pi \lambda_1 d^2 \beta^{\frac{2}{\alpha}}\rho),
\label{eq: tight_upper_lower_Y1}
\end{align}
where $(a)$ is obtained based on Proposition \ref{proposition: SIR_1}, by
replacing $\lambda_0$ with $\lambda_1$. Similarly, we can obtain that
\begin{equation}
\mathbb{E}\big[Y_2(h_{00}) \big] =  \exp(-\pi \lambda_1^c d^2
\gamma_1^{\frac{2}{\alpha}}\rho). \label{eq: tight_upper_lower_Y2}
\end{equation}

Finally, by substituting (\ref{eq: tight_upper_lower_Y1}) and (\ref{eq:
tight_upper_lower_Y2}) into the right-hand side of (\ref{eq: tight_upper_lower})
and then adopting the resulting  right-hand side of (\ref{eq:
tight_upper_lower}) to approximate $\mathcal{P}_{0}^{p} $, we can obtain a
tractable approximate to $\mathcal{P}_{0}^{p} $ for the case of $0<\gamma_1 <
\beta$ as in (\ref{eq: tractable_suc_prob}). Proposition \ref{proposition:
method_2} is thus proved.

\section{Proof to Proposition \ref{proposition: relationship_cn_cn-1}}
From (\ref{eq: C_p_N}), we have
\begin{align}
 \mathcal{C}^{p,N}=\lambda_0 \mathbb{P}\big(\textrm{SIR}_0^{(0)}\geq
\gamma_1,...,\textrm{SIR}_0^{(N-1)}\geq \gamma_N\big)
 \mathbb{P}\big( \textrm{SIR}_0^{(N)}\geq \beta\big| \textrm{SIR}_0^{(0)}\geq
\gamma_1,...,\textrm{SIR}_0^{(N-1)}\geq \gamma_N \big). \label{eq: C_p_N_con}
\end{align}
Thus, based on (\ref{eq: C_p_N}) and (\ref{eq: C_p_N_con}), we have
\begin{align}
\frac{ \mathcal{C}_0^{p,N}}{ \mathcal{C}_0^{p,N-1}}=
\frac{\mathbb{P} \big(\textrm{SIR}_0^{(0)}\geq \gamma_1, \cdots,
\textrm{SIR}_0^{(N-1)}\geq \gamma_{N} \big) }{\mathbb{P}
\big(\textrm{SIR}_0^{(0)}\geq \gamma_1, \cdots,   \textrm{SIR}_0^{(N-1)}\geq
\beta \big)}   \mathbb{P}\big(\textrm{SIR}_0^{(N)}\geq \beta \big|
\textrm{SIR}_0^{(0)}\geq \gamma_1, \cdots,  \textrm{SIR}_0^{(N-1)}\geq
\gamma_{N} \big). \label{eq: relationship_pn_pn-1}
\end{align}
Since both proposed schemes adopt the same SIR thresholds $\gamma_k$ for any
$k \in \{1,...,N-1\}$, the   distributions of these $\textrm{SIR}_0^{(k)}$'s are
the same for both proposed schemes and thus do not affect the ratio of $\frac{
\mathcal{C}_0^{p,N}}{ \mathcal{C}_0^{p,N-1}}$.
Hence, in the following, we  focus on the distribution $\textrm{SIR}_0^{(N)}$ by
varying $\gamma_{N}\in [0,\infty)$, and compare $\frac{ \mathcal{C}_0^{p,N}}{
\mathcal{C}_0^{p,N-1}}$ with $1$.
With a proof similar to that of Proposition \ref{proposition:
relationship_C1_C2}, it is easy to verify that
\begin{enumerate}
 \item if $0\leq \gamma_{N}\leq \gamma_{N-1}$, the distribution of
$\textrm{SIR}_0^{(N)}$ is the same as that of $\textrm{SIR}_0^{(N-1)}$; and thus
we have $\frac{ \mathcal{C}_0^{p,N}}{ \mathcal{C}_0^{p,N-1}}=1$;
 \item if $\gamma_{N}\geq \beta$, $\mathbb{P} \big(\textrm{SIR}_0^{(N)}\geq
\beta \big| \textrm{SIR}_0^{(0)}\geq \gamma_1, \cdots,
\textrm{SIR}_0^{(N-1)}\geq \gamma_{N} \big)=1$, and  $\mathbb{P}
\big(\textrm{SIR}_0^{(0)}\geq \gamma_1, \cdots,  \textrm{SIR}_0^{(N-1)}\geq
\gamma_{N} \big)\leq \mathbb{P} \big(\textrm{SIR}_0^{(0)}\geq \gamma_1, \cdots,
\textrm{SIR}_0^{(N)}\geq \beta \big)$, where ``$=$'' holds when $\gamma_{N} =
\beta$. Thus, from (\ref{eq: relationship_pn_pn-1}), if $\gamma_{N}> \beta$,
$\frac{ \mathcal{C}_0^{p,N}}{ \mathcal{C}_0^{p,N-1}}<1$, and if $\gamma_{N}=
\beta$, $\frac{ \mathcal{C}_0^{p,N}}{ \mathcal{C}_0^{p,N-1}}=1$;
 \item if $\gamma_{N-1}<\gamma_{N}< \beta$, we have
\begin{align*}
\frac{ \mathcal{C}_0^{p,N}}{ \mathcal{C}_0^{p,N-1}}&>
\frac{\mathbb{P} \big(\textrm{SIR}_0^{(0)}\geq \gamma_1, \cdots,
\textrm{SIR}_0^{(N-1)}\geq \gamma_{N} \big)}{\mathbb{P}
\big(\textrm{SIR}_0^{(0)}\geq \gamma_1, \cdots,   \textrm{SIR}_0^{(N-1)}\geq
\beta \big)} \\ \nonumber
&~~\times \mathbb{P} \big(\textrm{SIR}_0^{(N-1)}\geq \beta \big|
\textrm{SIR}_0^{(0)}\geq \gamma_1, \cdots,  \textrm{SIR}_0^{(N-1)}\geq
\gamma_{N} \big) =1.
\end{align*}  \vspace{-2.8mm}
\end{enumerate}
Proposition \ref{proposition: relationship_cn_cn-1} thus follows.

\end{document}